%% file: GH.tex
\documentclass[reqno]{amsart}

\usepackage[letterpaper,top=2cm,bottom=2cm,left=3cm,right=3cm,marginparwidth=1.75cm]{geometry}
\usepackage{amsmath}
\usepackage{amsthm}
\usepackage{amssymb}
\usepackage{amsbsy}
\usepackage{indentfirst}
\usepackage{fullpage}
\usepackage{latexsym}
\usepackage{mathtools}
\usepackage{eqparbox}
\usepackage{algorithm}
\usepackage{algpseudocode}
\usepackage{float}
\usepackage{cite}
\usepackage{nicefrac}
\usepackage{array}
\usepackage{multirow}
\usepackage{hhline}
\usepackage{setspace}
\usepackage{graphicx}
\usepackage{tikz}
\usepackage[colorlinks=true, allcolors=blue]{hyperref}

\definecolor{Xcolor}{HTML}{FF0049}
\definecolor{Ycolor}{HTML}{00BBEE}

\newcommand\MyBox[2]{
	\fbox{\lower0.75cm
		\vbox to 3cm{\vfil
			\hbox to 3cm{\hfil\parbox{1.4cm}{#1\\#2}\hfil}
			\vfil}%
	}%
}
\makeatletter
\newcommand\fs@ruled@notop{\def\@fs@cfont{\bfseries}\let\@fs@capt\floatc@ruled
	\def\@fs@pre{}%
	\def\@fs@post{\kern2pt\hrule\relax}%
	\def\@fs@mid{\kern2pt\hrule\kern2pt}%
	\let\@fs@iftopcapt\iftrue}
\renewcommand\fst@algorithm{\fs@ruled@notop}
\makeatother
\newcommand{\Input}[1]{\Statex \Comment{\parbox[t]{.978\linewidth}{\textbf{input:} #1}}}
\newcommand{\Output}[1]{\Statex \Comment{\parbox[t]{.978\linewidth}{\textbf{output:} #1}}}
\newcommand{\hrulestart}{\vspace{.5cm} \hrule}
\newcommand{\hruleend}{\hrule \vspace{1cm}}
\DeclareMathOperator*{\argmin}{arg\,min}

\DeclareMathOperator{\diam}{diam}

\DeclareMathOperator{\dis}{dis}

\newcommand{\defeq}{\buildrel \mathrm{def}\over =}
\newcommand{\N}{\mathbb{N}}

\newcommand{\R}{\mathbb{R}}
\newcommand{\MM}{\mathcal{M}}
\newcommand{\RR}{\mathcal{R}}

\newcommand{\dH}{d_{\mathcal{H}}}
\newcommand{\dlinf}{d_{l^\infty}}
\newcommand{\dGH}{d_{\mathcal{GH}}}
\newcommand{\dmGH}{\widehat{d}_{\mathcal{GH}}}
\newcommand{\K}{\mathcal{K}}

\newcommand{\DX}{D^X}
\newcommand{\DY}{D^Y}
\newcommand{\row}[1]{\mathrm{row}_{#1}}
\newcommand{\I}[1]{\langle #1 \rangle}
\newcommand{\ps}{\mathrm{PSPS}}
\newtheorem{theorem}{Theorem}

\newtheorem{lemma}{Lemma}
\newtheorem{claim}{Claim}
\theoremstyle{definition}

\theoremstyle{remark}
\newtheorem*{remark}{Remark}
\newtheorem*{rep@theorem}{\rep@title}
\newcommand{\newreptheorem}[2]{%
	\newenvironment{rep#1}[1]{%
		\def\rep@title{#2 \ref{##1}}%
		\begin{rep@theorem}}%
		{\end{rep@theorem}}}
\makeatother
\newreptheorem{theorem}{Theorem}
\newreptheorem{lemma}{Lemma}
\newreptheorem{claim}{Claim}
\numberwithin{equation}{section}

\title{Efficient estimation of the modified Gromov--Hausdorff distance between unweighted graphs}
\author{Vladyslav Oles, Nathan Lemons, Alexander Panchenko}

\begin{document}

\maketitle

\input{text/0_abstract}

\input{text/00_notation}

\input{text/1_intro}

\input{text/2_background}

\input{text/3_lower_bound}

\input{text/4_upper_bound}

\input{text/5_algorithm}

\input{text/6_computational_examples}

\input{text/7_conclusion}

\input{text/8_acknowledgements}

\bibliographystyle{alpha}
\bibliography{bib/GH, bib/networks, bib/domain, bib/optimization, bib/ML, bib/graph_matching}

\input{text/9_appendix}

\end{document}

%% file: text/0_abstract.tex
\begin{abstract}
Gromov--Hausdorff distances measure shape difference between the objects representable as compact metric spaces, e.g. point clouds, manifolds, or graphs. Computing any Gromov--Hausdorff distance is equivalent to solving an NP-hard optimization problem, deeming the notion impractical for applications. In this paper we propose a polynomial algorithm for estimating the so-called modified Gromov--Hausdorff (mGH) distance, a relaxation of the standard Gromov--Hausdorff (GH) distance with similar topological properties. We implement the algorithm for the case of compact metric spaces induced by unweighted graphs as part of Python library \verb|scikit-tda|, and demonstrate its performance on real-world and synthetic networks. The algorithm finds the mGH distances exactly on most graphs with the scale-free property. We use the computed mGH distances to successfully detect outliers in real-world social and computer networks.
\end{abstract}

%% file: text/00_notation.tex
\section*{List of notation}
\def\arraystretch{1.5}
\begin{table}[H]
\begin{tabular}{rp{0.8\textwidth}}
$\langle n \rangle$ & index set $\{1, \ldots, n\}$, for $n \in \N$.\\
$\lceil a \rceil$ & the ceiling of $a \in \R$.\\ 
$\lfloor a \rceil$ & the nearest integer to $a \in \R$.\\
$\|\mathbf{v}\|_\infty$ & $l_\infty$-norm of vector $\mathbf{v} = \renewcommand{\arraystretch}{1}\begin{bmatrix}v_1 & v_2 & \ldots & \end{bmatrix}$.\\
$\row{i}(A)$ & $i$-th row of matrix $A$: $\renewcommand{\arraystretch}{1}\begin{bmatrix}A_{i,1} & A_{i,2} & \ldots &\end{bmatrix}$.\\
$A_{(i)(j)}$ & matrix obtained from matrix $A$ by removing its $i$-th row and $j$-th column.\\
$\dlinf(A, B)$ & $l_\infty$-induced matrix distance: $\max_{i,j} \left|A_{i,j} - B_{i,j}\right|$.\\
$|S|$ & number of elements in set $S$.\\
$S^{\times n}$ & Cartesian product $\underbrace{S \times \ldots \times S}_{n \text{ times}}$, for $n \in \N$.\\
$S \to T$ & set of the mappings of $S$ into $T$.\\
$\dH^X(S, T)$ & the Hausdorff distance between the subsets $S, T$ of metric space $(X, d_X)$:\\& $\displaystyle \dH^X(S, T) \defeq \max \left\{\sup_{s \in S} \inf_{t \in T} d_X(s, t), \; \sup_{t \in T} \inf_{s \in S} d_X(s, t) \right\}$.\\
$\DX$ & distance matrix of metric space $(X, d_X)$, where $X = \{x_1, \ldots, x_{|X|}\}$:\\&$\DX_{i,j} \defeq d_X(x_i, x_j) \quad \forall i,j = 1, \ldots, |X|$.\\
$\MM$ & set of compact metric spaces.\\
$\diam X$ & diameter of metric space $(X, d_X) \in \MM$: $\displaystyle \sup_{x, x' \in X} d_X(x, x')$.
\end{tabular}
\end{table}
\def\arraystretch{1}
\newpage

%% file: text/1_intro.tex
\section{Introduction}
\subsection{Isometry-invariant distances between metric spaces}
The Gromov--Hausdorff (GH) distance, proposed by Gromov in \cite{gromov1981groups} (see also \cite{tuzhilin2016invented}), measures how far two compact metric spaces are from being isometric to each other. Since its conception four decades ago, the GH distance was mainly studied from a theoretical standpoint, as its computation poses an intractable combinatorial problem \cite{chazal2009gromov, memoli2007use, schmiedl2017computational}. Computing the GH distance can be thought of as an extension of the NP-hard Quadratic Bottleneck Assignment Problem (QBAP) that allows non-injective assignments in its search space. In particular, it was shown that approximating the distance up to a factor of $\sqrt{N}$ in the general case (or up to a factor of 3 for ultrametric spaces or metric trees) cannot be done in polynomial time unless $\text{P}=\text{NP}$ \cite{schmiedl2017computational, agarwal2018computing}. However, a polynomial-time algorithm for $\sqrt{N}$-approximation was given for unweighted metric trees in \cite{agarwal2018computing}, and a superlinear $\frac{5}{4}$-approximation algorithm for subsets of $\mathbb{R}^1$ was derived in \cite{majhi2019approximating}. In \cite{memoli2004comparing}, the GH distance was first considered for shape comparison, and several of its computationally motivated relaxations were presented since then.

Different variations of the Gromov--Wasserstein distance, a relaxation of the GH distance for metric measure spaces motivated by the optimal transport problem \cite{villani2003topics}, were proposed in \cite{sturm2006geometry} and in \cite{memoli2007use}, and further studied in e.g. \cite{memoli2009spectral, memoli2011gromov, peyre2016gromov}. The Gromov--Wasserstein distances are based on a generalization of bijective assignments, which are less expressive than the assignments allowed in the GH distance. Similarly to the GH distance, computing the Gromov--Wasserstein distance also requires solving a non-convex optimization problem. Recently, semidefinite relaxations of both the GH and Gromov--Wasserstein distances were studied in \cite{villar2016polynomial}. While allowing polynomial-time approximation, these relaxations admit distance 0 between non-isometric objects, losing the desired property of being a metric. Another result of potential interest from \cite{villar2016polynomial} is a feasible algorithm for an upper bound of the GH (and therefore the mGH) distance. In \cite{lipman2011conformal}, the authors define the conformal Wasserstein distance, inspired by the Gromov--Wasserstein distance. It is a metric on the isometry classes of Riemannian 2-manifolds that can be accurately approximated in polynomial time under some reasonable conditions.

In \cite{memoli2012some}, M\'emoli introduces the modified Gromov--Hausdorff (mGH) distance, another relaxation of the GH distance that preserves the property of being a metric on the isometry classes of compact metric spaces. Same as the GH distance, the mGH distance is based on assignments between the metric spaces that are allowed to be non-injective. It turns out that the two distances are not equivalent in the general case \cite{oles2022lipschitz}, although they are topologically equivalent within precompact (in the GH-distance sense) families of compact metric spaces \cite{memoli2012some}.

Directly computing the mGH distance still poses an intractable combinatorial  optimization of $O(N^N)$ time complexity (as compared to $O(N^{2N})$ for the standard GH distance). Our focus on the mGH distance in this paper is partially motivated by the so-called ``structural theorem'' (Theorem 5.1 in \cite{memoli2012some}), which allows for the decomposition of the computation into solving a sequence of polynomial-time problems.

\subsection{Shape-based graph matching}
Because graphs are ubiquitous in applications, the task of graph matching, i.e. measuring how much a pair of graphs are different from each other, is extensively studied. Common approaches to exact graph matching are those based on graph shape, such as subgraph isomorphism and maximum common subgraph \cite{conte2004thirty, foggia2014graph}. 
The shape-based approaches appear in many fields including neuroscience \cite{van2010comparing}, telecommunications \cite{shoubridge2002detection}, and chemoinformatics \cite{raymond2002maximum, vandewiele2012genesys}. While efficient heuristics for these approaches exist for special cases (e.g. planar graphs), applying them in the general case requires solving an NP-complete problem \cite{bunke1997relation, conte2004thirty}.

Recently, the Gromov--Hausdorff framework for graph matching was explored both theoretically \cite{aflalo2015convex} and in applications, e.g. in the fields of neuroscience \cite{lee2011computing, lee2006persistent, hendrikson2016using}, social sciences, and finance \cite{hendrikson2016using}. Undirected graphs admit metric space representation using the geodesic (shortest path length) distances on their vertex sets. However, the high computational cost of computing the isometry-invariant distances impedes a more widespread application of this approach.

\subsection{Our contribution}
Our main contribution is a theoretical framework for producing polynomial-time lower bounds of the GH and mGH distances. Furthermore, we present an algorithm for estimating the mGH distance, built upon this framework. We implement the algorithm for metric representations of unweighted graphs, leveraging their properties to reduce the polynomial order in the algorithm's time complexity. We demonstrate the performance of the algorithm on several datasets of real-world and synthesized graphs of up to several thousand vertices.

The rest of the paper is structured as follows. Section \ref{background} briefly reviews \cite{memoli2013gromov} to formally define the GH and mGH distances, show their relation to each other, and state some of their properties. In Sections \ref{lower bound} and \ref{upper bound} we discuss the ideas for establishing lower and upper bounds, respectively, of the mGH distance between finite metric spaces. In Section \ref{algorithm}, we describe the algorithm for estimating the mGH distance, show that it has polynomial time complexity, then discuss and present its implementation for the case of unweighted graphs. Computational examples from real-world and synthetic datasets are given in Section \ref{computational examples}, and Section \ref{conclusion} summarizes our work. The Appendix contains pseudocode for the procedures and algorithms, omitted from the main paper for brevity.

%% file: text/2_background.tex
\section{Background}
\label{background}
When talking about metric space given by set $X$ and distance function $d_X: X \times X \to \R$, we use notation $(X, d_X)$ and its shorter version $X$ interchangeably. We expect the distinction between a set $X$ and a metric space $X$ to be clear from the context.
\subsection{Definition of the Gromov--Hausdorff distance}
Given $(X, d_X), (Y, d_Y) \in \MM$, where $\MM$ denotes the set of all compact metric spaces, the GH distance measures how far the two metric spaces are from being isometric. It considers any ``sufficiently rich'' third metric space $(Z, d_Z)$ that contains isometric copies of $X$ and $Y$, measuring the Hausdorff distance (in $Z$) between these copies, and minimizes over the choice of the isometric copies and $Z$. Formally, the GH distance is defined as $$\dGH(X, Y) \defeq \inf_{Z, \phi_X, \phi_Y}\dH^Z\big(\phi_X(X), \phi_Y(Y)\big),$$ where $\phi_X : X \to Z$ and $\phi_Y : Y \to Z$ are isometric embeddings of $X$ and $Y$ into $Z$, and $\dH^Z$ is the Hausdorff distance in $Z$: $$\dH^Z(S, T) \defeq \max \Big\{\sup_{s \in S} \inf_{t \in T} d_Z(s, t),\sup_{t \in T} \inf_{s \in S} d_Z(s, t) \Big\}\quad \forall S, T \subseteq Z$$ (see Figure \ref{fig:dGH}). Gromov has shown in \cite{gromov2007metric} that $\dGH$ is a metric on the isometry classes of $\MM$, constituting what is called a Gromov--Hausdorff space.

\begin{figure}[b]
    \centering
\caption{Illustration of the idea underlying the Gromov--Hausdorff distance.}
\label{fig:dGH}
\vspace{.2cm}
\begin{tikzpicture}[scale=1.3]
\begin{scope}[xshift=-65]
\draw[fill=Xcolor] plot[smooth cycle] coordinates {(-2.5,0.5) (-1.5, .5) (-1.5,-0.5) (-3.25, -.5) (-2,0)} node[black,xshift=-25,yshift=5] {$X$};
\end{scope}
\begin{scope}[xshift=-10,yshift=-5,rotate=-30]
\draw[fill=Ycolor] (0,0) ellipse (1.5 and .25) node[black,xshift=15,yshift=15] {$Y$};
\end{scope}
\draw[black,dashed] plot coordinates {(-2.5,-.7)  (-6,-1.25) (-4,-4) (-1, -4) (1,-3) (.6, -2) (-2.5,-.7)} node[black,below] {$Z$};
\begin{scope}[xshift=-40,yshift=-125, rotate=-50]
\draw[fill=Xcolor,opacity=.3,text opacity=1] plot[smooth cycle] coordinates {(-2.5,0.5) (-1.5, .5) (-1.5,-0.5) (-3.25, -.5) (-2,0)} node[black,xshift=-43,yshift=-5] {$\phi_X(X)$};
\end{scope}
\begin{scope}[xshift=-50,yshift=-80,rotate=20]
\draw[fill=Ycolor,opacity=.3,text opacity=1] (0,0) ellipse (1.5 and .25) node[black,xshift=23,yshift=-12] {$\phi_Y(Y)$};
\end{scope}
\draw[->] (-4,-.7)--(-3.3,-2.3) node[midway,left]{$\phi_X$};
\draw[->] (-.6,-.5)--(-.9,-2.1) node[midway,right]{$\phi_Y$};

\end{tikzpicture}

\end{figure}

Although the above definition gives the conceptual understanding of the GH distance, it is not very helpful from the computational standpoint. The next subsection introduces a more practical characterization of the GH distance.
\subsection{Characterization of the GH distance}
For two sets $X$ and $Y$, we say that relation $R \subseteq X \times Y$ is a \textit{correspondence} if for every $x \in X$ there exists some $y \in Y$ s.t. $(x, y) \in R$ and for every $y \in Y$ there exists some $x \in X$ s.t. $(x, y) \in R$. We denote the set of all correspondences between $X$ and $Y$ by $\RR(X, Y)$. 

If $R$ is a relation between metric spaces $(X, d_X)$ and $(Y, d_Y)$, its \textit{distortion} is defined as the number $$\dis R \defeq \sup_{(x, y), (x', y') \in R} \big|d_X(x, x') - d_Y(y, y')\big|.$$ Notice that any mapping $\varphi: X \to Y$ induces the relation $R_\varphi \defeq \big\{\big(x, \varphi(x)\big): x \in X \big\}$, and we denote $$\dis \varphi \defeq \dis R_\varphi = \sup_{x, x' \in X} \big|d_X(x, x') - d_Y\big(\varphi(x), \varphi(x')\big)\big|.$$ Similarly, any $\psi: Y \to X$ induces the relation $R_\psi \defeq \big\{\big(\psi(y), y\big): y \in Y \big\}$. If both $\varphi: X \to Y$ and $\psi: Y \to X$ are given, we can define the relation $R_{\varphi, \psi} \defeq R_\varphi \cup R_\psi$, and realize that it is actually a correspondence, $R_{\varphi, \psi} \in \RR(X, Y)$.

A useful result in \cite{kalton1999distances} identifies computing GH distance with solving an optimization problem, either over the correspondences between $X$ and $Y$ or over the functions $\varphi: X \to Y$ and $\psi: Y \to X$: $$\dGH(X, Y) = \frac{1}{2}\inf_{R \in \RR(X, Y)} \dis R = \frac{1}{2} \inf_{\varphi, \psi}\dis R_{\varphi, \psi}.$$


By definition, distortion of any relation $R \subseteq X \times Y$ is bounded by $\dis R \leq d_{\max}$, where $d_{\max} \defeq \max\{\diam X, \diam Y\}$. Combined with the characterization of the GH distance, it yields $\dGH(X, Y) \leq \frac{1}{2} d_{\max}$.

Let $*$ denote the (compact) metric space that is comprised of exactly one point. For any correspondence $R \in \RR(X, *)$, $\dis R = \sup_{x, x' \in X} \big|d_X(x, x') - 0\big| = \diam X$. The above characterization implies that $\dGH(X, *) = \frac{1}{2}\diam X$, and, by an analogous argument, $\dGH(Y, *) = \frac{1}{2}\diam Y$. From the triangle inequality for the GH distance, $\dGH(X, Y) \geq \big|\dGH(X, *) - \dGH(Y, *)\big| = \frac{1}{2}|\diam X - \diam Y|$.

\subsection{Modifying the GH distance}
Recall that for some $\varphi: X \to Y$ and $\psi: Y \to X$, correspondence $R_{\varphi, \psi}$ is defined as $R_{\varphi, \psi} = R_\varphi \cup R_\psi$. For any two elements in $R_{\varphi, \psi}$, either both belong to $R_\varphi$, or both belong to $R_\psi$, or one of them belongs to $R_\varphi$ while the other belongs to $R_\psi$. It follows that $$\dis R_{\varphi, \psi} = \max \big\{\dis R_\varphi, \dis R_\psi, C_{\varphi, \psi}\big\},$$ where $C_{\varphi, \psi} \defeq \displaystyle \sup_{x \in X, y \in Y} \big|d_X\big(x, \psi(y)\big) - d_Y\big(\varphi(x), y\big)\big|$.

Notice that the number $C_{\varphi, \psi}$ acts as a coupling term between the choices of $\varphi$ and $\psi$ in the optimization problem $$\dGH(X, Y) = \frac{1}{2}\inf_{\varphi, \psi}\max \big\{\dis R_\varphi, \dis R_\psi, C_{\varphi, \psi}\big\},$$ making its search space to be of the size $|X|^{|Y|}|Y|^{|X|}$. Discarding the coupling term $C_{\varphi, \psi}$ yields the notion of \textit{the modified Gromov--Hausdorff distance} $$\dmGH(X, Y) \defeq \frac{1}{2}\inf_{\varphi, \psi} \max \big\{\dis R_\varphi, \dis R_\psi\big\} \leq \dGH(X, Y).$$ Computing $\dmGH(X, Y)$ requires solving two decoupled optimization problems whose search spaces are of the size $|X|^{|Y|}$ and $|Y|^{|X|}$, respectively. An equivalent definition emphasizing this fact is given by $$\dmGH(X, Y) \defeq \frac{1}{2}\max \big\{\inf_{\varphi} \dis R_\varphi, \inf_{\psi} \dis R_\psi\big\}.$$

Similarly to $\dGH$, $\dmGH$ is a metric on the isometry classes of $\MM$. Moreover, $\dmGH$ is topologically equivalent to $\dGH$ within GH-precompact families of metric spaces \cite{memoli2012some}.

\subsection{Curvature sets and the structural theorem}
Let $X \in \MM$, and consider $(x_1, \ldots, x_n) \in X^{\times n}$, an $n$-tuple of points in $X$ for some $n \in \N$. The $n \times n$ matrix containing their pairwise distances is called the \textit{distance sample} induced by $(x_1, \ldots, x_n)$, and denoted by $D^{(x_1, \ldots, x_n)} \defeq \big(d_X(x_i, x_j)\big)_{i,j=1}^n$. A distance sample generalizes the notion of distance matrix of $\{x_1, \ldots, x_n\}$ when $x_1, \ldots, x_n$ are not necessarily distinct. Unlike a distance matrix, a distance sample may contain zeros off the main diagonal.

The \textit{$n$-th curvature set} of $X$ is then defined as a set of all $n \times n$ distance samples of $X$, denoted $$\K_n(X) \defeq \left\{D^{(x_1, \ldots, x_n)}: (x_1, \ldots, x_n) \in X^{\times n}\right\}.$$ For example, $\K_2(X)$ contains the same information as the entries of $D^X$, the distance matrix of $X$; when $X$ is a smooth planar curve, then its curvature at every point can be calculated from the information contained in $\K_3(X)$ \cite{calabi1998differential}, thus providing motivation for the name \textit{curvature sets}.

Curvature sets contain information about the shape of a compact metric space in permutation-invariant way. In particular, any $X \in \MM$ and $Y \in \MM$ are isometric if and only if $\K_n(X) = \K_n(Y)$ for every $n \in \N$ (3.27 in \cite{gromov2007metric}). To discriminate the shapes of $X$ and $Y$, it is therefore reasonable to measure the difference between $\K_n(X)$ and $\K_n(Y)$ for various $n \in \N$. Since both $n$-th curvature sets are subsets of the same space $\R^{n \times n}$, Hausdorff distance is a natural metric between them. We equip the set of $n \times n$ matrices with distance $\dlinf(A, B) \defeq \max_{i,j} \left|A_{i,j}-B_{i,j}\right|$, and define $$d_{\K_n}(X, Y) \defeq \frac{1}{2}\dH^{\R^{n \times n}}\big(\K_n(X), \K_n(Y)\big),$$ where $\dH^{\R^{n \times n}}$ is the $\dlinf$-induced Hausdorff distance on $\R^{n \times n}$.

\begin{remark}
	The choice of distance $\dlinf: \R^{n \times n} \times \R^{n \times n} \to \R$ complies with the notion of distortion of a mapping. If $\varphi$ is a mapping from $X$ to $Y$ for $X = \{x_1, \ldots, x_{|X|}\}$, $Y$ --- metric spaces, then $$\dis \varphi = \dlinf\left(D^X, D^{\left(\varphi(x_1), \ldots, \varphi(x_{|X|})\right)}\right).$$ The fact that $\varphi$ can be non-injective provides intuition for the possibility of identical points in a tuple from the definition of a distance sample.
\end{remark}

An important result, extensively relied upon in this paper, is the so-called ``structural theorem'' for the mGH \cite{memoli2012some, memoli2013gromov}: $$\dmGH(X, Y) = \sup_{n \in \N}d_{\K_n}(X, Y).$$

Notice that the bounds of the GH distance from the inequalities $\frac{1}{2} |\diam X - \diam Y| \leq \dGH(X, Y) \leq \frac{1}{2}d_{\max}$ also hold for the mGH distance: 
\begin{align*}
\dmGH(X, Y) &\geq d_{\K_2}(X, Y)
\\ &= \frac{1}{2}\dH^\R\Big(\big\{d_X(x, x'): x, x' \in X\big\},\big\{d_Y(y, y'): y, y' \in Y\big\}\Big)
\\ &\geq \frac{1}{2}\max \bigg\{\inf_{x, x' \in X}\big|\diam Y - d_X(x, x')\big|, \inf_{y, y' \in Y}\big|\diam X - d_Y(y, y')\big|\bigg\}
\\ &= \frac{1}{2}\left|\diam X - \diam Y\right|,
\end{align*}
while $\dmGH(X, Y) \leq \dGH(X, Y) \leq d_{\max}$ trivially follows from the definition of the mGH distance.

%% file: text/3_lower_bound.tex
\section{Lower bound for $\dmGH(X, Y)$}
\label{lower bound}
This section provides theoretical results and algorithms for an informative lower bound for the mGH distance between a pair of metric spaces $X$ and $Y$. This and the following sections assume that the metric spaces are non-empty and finite, i.e. $1 \leq |X|, |Y| < \infty$. When talking about algorithmic time complexities, we denote the input size with $N \defeq \max\{|X|, |Y|\}$. 

\begin{remark}
We notice that efficient estimation of either the GH or the mGH distance between finite metric spaces allows for its estimation between infinite compact metric spaces with desired accuracy. This is because for any $X \in \MM$ and $\epsilon > 0$, there exist finite $X_\epsilon \subset X$ that is an $\epsilon$-net of $X$, which gives $$\dmGH(X, X_\epsilon) \leq \dGH(X, X_\epsilon) \leq \dH(X, X_\epsilon) \leq \epsilon.$$ Similarly, for any $Y \in \MM$ there exists its finite $\epsilon$-net $Y_\epsilon$, and therefore $$|\dGH(X, Y) - \dGH(X_\epsilon, Y_\epsilon)| \leq \dGH(X, X_\epsilon) + \dGH(Y, Y_\epsilon) \leq 2\epsilon.$$


\end{remark}

Feasible algorithms for lower bounds are important in e.g. classification tasks, where knowledge that a distance exceeds some threshold can make computing the actual distance unnecessary. In particular, if the mGH distance between metric representations of two graphs is $> 0$, it immediately follows that the graphs are not isomorphic.

\subsection{$d$-bounded matrices}
Let $A$ be a square matrix. We say that $A$ is \textit{$d$-bounded} for some $d \in \R$ if every off-diagonal entry of $A$ is $\geq d$. Similarly, $A$ is \textit{positive-bounded} if its off-diagonal entries are positive. Naturally, any $d$-bounded matrix for $d > 0$ is also positive-bounded.

Notice that a distance sample $D^{(x_1, \ldots, x_n)}$ of $X$ is $d$-bounded if and only if $d_X(x_i, x_j) \geq d \quad \forall i \neq j$, and positive-bounded if and only if $x_1, \ldots, x_n$ are distinct. By non-negativity of a metric, any distance sample is 0-bounded.

\begin{claim}
	Let $A$ and $B$ be square matrices of the same size. If $A$ is $d$-bounded for some $d > 0$, and $B$ is 0-bounded but not positive-bounded, then $\dlinf(A, B) \geq d$.
	\label{claim 1}
\end{claim}

Recall that a matrix $B$ is a permutation similarity of a (same-sized) matrix $A$ if $B = PAP^{-1}$ for some permutation matrix $P$. Equivalently, $B$ is obtained from $A$ by permuting both its rows and its columns according to some permutation $\pi$: $B_{i,j} = A_{\pi(i),\pi(j)}$. Given $n \in \N$, we will denote the set of permutation similarities of $n \times n$ principal submatrices of $A$ by $\ps_n(A)$.

\begin{claim}
	A distance sample $K \in \K_n(X)$ is positive-bounded if and only if it is a permutation similarity of a principal submatrix of $\DX$, i.e. if and only if $K \in \ps_n(\DX)$. In particular, there are no positive-bounded distance samples in $\K_n(X)$ if $n > |X|$.
	\label{claim 2}
\end{claim}


\begin{theorem}
	Let $K \in \K_n(X)$ be $d$-bounded for some $d > 0$. If $n > |Y|$, then $\dmGH(X, Y) \geq \frac{d}{2}$.
	\label{theorem 1}
\end{theorem}
\begin{proof}
	Notice that $L \in \K_n(Y)$ implies that $L$ is 0-bounded (by non-negativity of a metric) and not positive-bounded (from Claim \ref{claim 2}). Then
	\begin{align*}
	\dmGH(X, Y) &\geq \frac{1}{2}\dH^{\R^{n \times n}}\big(\K_n(X), \K_n(Y)\big) 
	\\ &\geq \frac{1}{2}\min_{L \in \K_n(Y)} \dlinf(K, L)
	\\ &\geq\frac{d}{2}.\hspace{4cm}\text{from Claim \ref{claim 1}}
	\end{align*}

\end{proof}

\begin{remark}
    For the standard GH distance, its lower bound implied by Theorem \ref{theorem 1} is a direct consequence of its relationship with the $n$-packing number established in \cite{cho1997optimal}.
\end{remark}

\subsection{Obtaining $d$-bounded distance samples of large size}
In order to apply Theorem \ref{theorem 1} for some $d > 0$, one needs to verify the existence of a $d$-bounded distance sample from $X$ that exceeds $Y$ in size. Ideally, one wants to know $M(X, d)$, the largest size of a $d$-bounded distance sample of $X$: $$M(X, d) \defeq \max \big\{n \in \N: \text{$\exists$ $d$-bounded $K \in \K_n(X)$} \big\}.$$ Equivalently, $M(X, d)$ is the so-called $d$-packing number of $X$: the largest number of points one can sample from $X$ such that they all are at least $d$ away from each other. Finding $M(X, d)$ is equivalent to finding the size of a maximum independent set of the graph $G = \big(X, \{(x_i, x_j): d_X(x_i, x_j) < d\}\big)$. Unfortunately, this problem is known to be NP-hard \cite{kegl2002intrinsic}, and we therefore require approximation techniques to search for a sufficiently large $d$-bounded distance sample of $X$.

We implement greedy algorithm \textproc{FindLarge$K$} (see Appendix \ref{FindLarge$K$} for the pseudocode) that, given the distance matrix of $X$ and some $d > 0$, finds in $O(N^3)$ time a $d$-bounded distance sample $K \in \K_{\widetilde{M}(X, d)}(X)$, where $\widetilde{M}(X, d)$ is an approximation of $M(X, d)$. Informally, the algorithm iteratively removes rows (and same-index columns) from $\DX$ until all off-diagonal entries of the resulting $K$ are $\geq d$. At each step, the algorithm chooses to remove a row that is \textit{least $d$-bounded}, i.e. one with the highest (non-zero) count of off-diagonal entries $< d$.

Removing a row from $K$ also ``collaterally'' removes an entry $< d$ from some of the remaining rows (due to removal of the corresponding column), potentially turning them into $d$-bounded. The counts of off-diagonal entries $<d$ in these rows define the number of newly appeared $d$-bounded rows in the remaining matrix at each step, thus controlling the duration of \textproc{FindLarge$K$}. In particular, the maximum duration (and therefore the worst approximation $\widetilde{M}(X, d)$) is attained when ``collaterally'' removed entries $< d$ reside in the least $d$-bounded rows at each step, which minimizes the number of $d$-bounded rows obtained from the removal. Let $w$ be the maximum count of off-diagonal entries $<d$ in the rows of $\DX$. Assuming the worst case scenario, each iteration of $\textproc{FindLarge$K$}(\DX, d)$ reduces the number of rows with $w$ off-diagonal entries $<d$ by $w+1$ (one row is removed and the counts for $w$ others become $w-1$), or to zero if fewer than $w+1$ of them remains. For the sake of simplicity, we will assume that $w+1$ divides $|X|$. It follows that after $\frac{|X|}{w+1}$ iterations the maximum count of off-diagonal entries $<d$ in the $\frac{w}{w+1}|X|$ rows of $K$ is at most $w-1$. More generally, if for some $k \leq w$ the current $K$ is comprised by $\big(1-\frac{k}{w+1}\big)|X|$ rows with at most $w-k$ off-diagonal entries $<d$ each, then after another $$\frac{\big(1-\frac{k}{w+1}\big)|X|}{w-k+1} = \frac{|X|}{w+1}$$
iterations the maximum count will be at most $w-k-1$. By induction, $K$ will contain no off-diagonal entries $<d$ after at most $w\frac{|X|}{w+1}$ iterations, which implies that its resulting size $\widetilde{M}(X, d)$ must be at least $\frac{1}{w+1}|X|$. Because $M(X, d) \leq |X|$, the approximation ratio of $\widetilde{M}(X, d)$ is $\frac{1}{w+1}$.

Notice that the obtained $K$ does not need to be unique, because at any step there can be multiple least $d$-bounded rows. Choosing which one of them to remove allows selecting for some desired characteristics in the retained rows of $K$. In particular, subsection \ref{verifying lower bound} provides motivation to select for bigger entries in the resulting $d$-bounded distance sample $K$. To accommodate this, we choose to remove at each step a least $d$-bounded row with the smallest sum of off-diagonal entries $\geq d$, so-called \textit{smallest least $d$-bounded} row. Since uniqueness of a smallest least $d$-bounded row is not guaranteed either, we implemented procedure \textproc{FindLeastBoundedRow} (see Appendix \ref{FindLeastBoundedRow} for the pseudocode) to find the index of the first such row in a matrix.

\subsection{Using permutation similarities of principal submatrices of $\DY$}
\label{principal submatrices}
Theorem \ref{theorem 1} establishes that $\dmGH(X, Y) \geq \frac{d}{2}$ from the existence of some $d$-bounded distance sample of sufficiently large size. However, such distance samples may not exist in certain cases (e.g. when $|X| = |Y|$), thus deeming Theorem \ref{theorem 1} inapplicable. The main result of this subsection, Theorem \ref{theorem 2}, complements Theorem \ref{theorem 1} by allowing to verify $\dmGH(X, Y) \geq \frac{d}{2}$ in those cases.

\begin{lemma}
	Let $K \in \K_n(X)$ be $d$-bounded for some $d > 0$, and let $n \leq |Y|$. If for some $i \in \I{n} \defeq \{1, \ldots, n\}$ \begin{align*}
	    \|\row{i}(K) - \row{i}(L)\|_\infty \geq d \quad \forall L \in \ps_n(\DY),
	\end{align*} then $\dlinf\big(K, \K_n(Y)\big) \geq d$.
	\label{lemma 1}
\end{lemma}
\begin{proof}
	Let $L \in \K_n(Y)$. If $L$ is not positive-bounded, $\dlinf(K, L) \geq d$ follows from Claim \ref{claim 1} and the fact that any distance sample is 0-bounded. If $L$ is positive-bounded, then $L \in \ps_n(\DY)$ from Claim \ref{claim 2}, and the premise entails
	\begin{align*}
	\dlinf(K, L) \geq \big\|\row{i}(K) - \row{i}(L)\big\|_\infty \geq d.
	\end{align*}
	Because $\dlinf(K, L) \geq d$ for an arbitrary choice of $L \in \K_n(Y)$, we have $\dlinf\big(K, \K_n(Y)\big) \geq d.$
\end{proof}

\begin{remark}
	A naive approach to proving $\dlinf(K, \K_n(Y)) \geq d$ is to show that $\dlinf(K, L) \geq d$ for each $L \in \K_n(Y)$, which comprises an instance of the NP-hard QBAP. Instead, the premise of Lemma \ref{lemma 1} (for a particular $i$) can be checked by solving at most $|Y|$ optimization problems of $O(|Y|)$ time complexity each, as will be shown in the next subsection.
\end{remark}

\begin{theorem}
	Let $K \in \K_n(X)$ be $d$-bounded for some $d > 0$, and let $n \leq |Y|$.  If for some $i \in \I{n}$ $$\|\row{i}(K) - \row{i}(L)\|_\infty \geq d \quad \forall L \in \ps_n(\DY),$$ then $\dmGH(X, Y) \geq \frac{d}{2}$.
	\label{theorem 2}
\end{theorem}
\begin{proof}
	\begin{align*}
	\dmGH(X, Y) &\geq \frac{1}{2}\dH^{\R^{n \times n}}\big(\K_n(X), \K_n(Y)\big) 
	\\ &\geq \frac{1}{2}\dlinf\big(K, \K_n(Y)\big) 
	\\ &\geq \frac{d}{2}. \hspace{4cm}\text{from Lemma \ref{lemma 1}}
	\end{align*}
\end{proof}


\subsection{Verifying $\dmGH(X, Y) \geq \frac{d}{2}$}
\label{verifying lower bound}
Let $d > 0$. To see if we can verify $\dmGH(X, Y) \geq \frac{d}{2}$ from $\DX$ and $\DY$, we start by calling \textproc{FindLarge$K$} to obtain a $d$-bounded distance sample $K \in \K_n(X)$ for some $n \geq \frac{1}{w+1}M(X, d)$. If $n > |Y|$, then $\dmGH(X, Y) \geq \frac{d}{2}$ follows immediately from Theorem \ref{theorem 1}. Otherwise, we want to obtain this lower bound from Theorem \ref{theorem 2}, which requires showing that some $i \in \I{n}$ satisfies $$\|\row{i}(K) - \row{i}(L)\|_\infty \geq d \quad \forall L \in \ps_n(\DY).$$

Let $i$ be fixed. If $L \in \ps_n(\DY)$, then all entries in $\row{i}(L)$ come from one row of $\DY$, with $L_{i,i} = 0$ being the diagonal element in that row. The choice of $i$ thus induces a (disjoint) partition of $\ps_n(\DY)$: $$\ps_n(\DY) = \bigcup_{j=1}^{|Y|} \ps_n^{i \gets j}(\DY),$$ where $\ps_n^{i \gets j}(\DY)$ is the set of all permutation similarities of principal submatrices of $\DY$ whose $i$-th row is comprised of the entries in $\row{j}(\DY)$. Therefore, the condition $$\|\row{i}(K) - \row{i}(L)\|_\infty \geq d \quad \forall L \in \ps_n(\DY)$$ can be verified by showing that $$\|\row{i}(K) - \row{i}(L)\|_\infty \geq d \quad \forall L \in \ps_n^{i\gets j}(\DY), \forall j \in \I{|Y|}.$$

Let $j$, in addition to $i$, be fixed. Notice that any $L \in \ps_n^{i\gets j}(\DY)$ corresponds to an injective mapping $f_L:\I{n} \to \I{|Y|}$ that defines the entries from $\row{j}(\DY)$ that populate $\row{i}(L)$: $L_{i, k} = \DY_{j, f_L(k)}$. In particular, $f_L(i) = j$, because $L_{i,i} = \DY_{j,j}$ for any $L \in \ps_n^{i\gets j}(\DY)$. Therefore, the condition $$\|\row{i}(K) - \row{i}(L)\|_\infty \geq d \quad \forall L \in \ps_n^{i\gets j}(\DY)$$ is equivalent to the non-existence of an injective $f_L:\I{n} \to \I{|Y|}$ such that $|K_{i, k} - \DY_{j, f_L(k)}| < d \quad \forall k \in \I{n}$ and $f_L(i) = j$. The decision about the existence of a feasible assignment $f_L$ between the entries of $\row{i}(K)$ and $\row{j}(\DY)$ is an instance of linear assignment feasibility problem. If such $f_L$ exists, it can be constructed by iteratively pairing the smallest unassigned $K_{i, k}$ to the smallest available $\DY_{j, h}$ s.t. $|K_{i, k} - \DY_{j, h}| < d$, that is, by setting $f_L(k) = h$ (see Claim \ref{feasible assignment}). 
It follows that each entry in $\row{j}(\DY)$ needs to be checked at most once, and hence solving the feasibility problem takes $O(|Y|)$ time if the entries in both $\row{i}(K)$ and $\row{j}(\DY)$ are sorted. Procedure \textproc{SolveFeasibleAssignment} (see Appendix \ref{SolveFeasibleAssignment} for the pseudocode) implements the solution for a pair of vectors, given their entries are arranged in ascending order. We notice that ignoring the actual order of the (off-diagonal) entries in $\row{i}(K)$ and $\row{j}(\DY)$ reflects the fact that curvature sets are closed under permutations of the underlying tuples of points.

\begin{claim}
    Let $\mathbf{v} \in \mathbb{R}^p, \mathbf{u} \in \mathbb{R}^q$ be s.t. $v_1 \leq \ldots \leq v_p, u_1 \leq \ldots \leq u_q$. If $\textproc{SolveFeasibleAssignment}(\mathbf{v}, \mathbf{u}, d)=\mathrm{FALSE}$, then no injective $f: \I{p} \to \I{q}$ s.t. $|v_t - u_{f(t)}| < d \quad \forall t \in \I{p}$ exists.
    \label{feasible assignment}
\end{claim}
\begin{proof}
    Let $t$ be the largest index for which $u_t$ could be feasibly assigned to some available entry of $\mathbf{v}$ by $\textproc{SolveFeasibleAssignment}(\mathbf{v}, \mathbf{u}, d)$. If such $t$ does not exist, then $|v_1 - u_l| \geq d \quad \forall l \in \I{q}$ and trivially no feasible assignment between $\mathbf{v}$ and $\mathbf{u}$ can exist.
    
    Denote the partially constructed assignment by $f: \I{t} \to \I{q}$, and notice that it is monotone by construction. In particular, every $u_l$ for $l > f(t)$ must remain available for assignments. Moreover, because $v_{t+1}$ could not be feasibly assigned to any entry of $\mathbf{u}$ after $u_{f(t)}$, which means that $v_t \leq u_l - d \quad \forall l > f(t-1)$.
    
    Let $r \leq t$ be the smallest index such that $f(r), f(r+1), \ldots, f(t)$ are a contiguous block, i.e. that $f(r) = f(t) - (t - r + 1)$. Because $u_{f(r)-1}$ remains available (or does not exist due to $f(r)=1$), $v_r$ could not be feasibly assigned to any entry of $\mathbf{u}$ before $u_{f(r)}$, which means that $v_r \geq u_l + d  \quad \forall l < f(r)$. 
    
    Because $v_r \leq \ldots \leq v_t \leq v_{t+1}$, none of these $t-r+2$ entries can be feasibly assigned to the entries of $\mathbf{u}$ before $u_{f(r)}$ or after $u_{f(t)}$, meaning that they only have $t-r+1$ feasible matches in $\mathbf{u}$. By the pigeonhole principle, no feasible assignment between the two vectors can exist.
\end{proof}

\begin{remark}
	Intuitively, either sufficiently small or sufficiently large entries in $\row{i}(K)$ for some $i \in \I{n}$ can render a feasible assignment $f_L$ impossible for every $j \in \I{|Y|}$, yielding $\|\row{i}(K) - \row{i}(L)\|_\infty \geq d \quad \forall L \in \ps_n(\DY)$ and therefore $\dmGH(X, Y) \geq \frac{d}{2}$ from Theorem \ref{theorem 2}. This provides the motivation behind considering the magnitude of the $\geq d$ entries when choosing a row to remove at each step of \textproc{FindLarge$K$}. Recall that such row is chosen by the auxiliary procedure \textproc{FindLeastBoundedRow}, that selects for bigger entries in the resulting $K$. The approach allows for a tighter lower bound in the case when the entries in $\DX$ are, generally speaking, bigger than those in $\DY$. The converse case is covered by similarly obtaining and handling a $d$-bounded distance sample of $Y$.
\end{remark}

Calling \textproc{SolveFeasibleAssignment}$(\row{i}(K), \row{j}(\DY), d)$ for every $j \in \I{|Y|}$ is sufficient to check whether a particular $i$ satisfies $\|\row{i}(K) - \row{i}(L)\|_\infty \geq d \quad \forall L \in \ps_n(\DY)$. Procedure \textproc{CheckTheoremB} (see Appendix \ref{CheckTheoremB} for the pseudocode) makes such check for each $i \in \I{n}$ to decide whether $\dmGH(X, Y) \geq \frac{d}{2}$ follows from Theorem \ref{theorem 2} for the $d$-bounded $K$. The procedure sorts the entries in the rows of $K$ and $\DY$ prior to the checks, which takes $O(N\log N)$ time for each of the $O(N)$ rows. This allows solving each of the $O(N^2)$ feasibility problems in $O(N)$ time, making the time complexity of \textproc{CheckTheoremB} $O(N^2\log N + N^3) = O(N^3)$.

Notice that both Theorem \ref{theorem 1} and \ref{theorem 2} regard a largest-size $d$-bounded distance sample of only one metric space, $X$. However, its counterpart for $Y$ is equally likely to provide information for discriminating the two metric spaces. Making use of the symmetry of $\dmGH$, we summarize theoretical findings of this section under $O(N^3)$-time procedure \textproc{VerifyLowerBound} (see Appendix \ref{VerifyLowerBound} for the pseudocode), that attempts to prove $\dmGH(X, Y) \geq \frac{d}{2}$.

\subsection{Obtaining the lower bound}
Procedure \textproc{VerifyLowerBound} is a decision algorithm that gives a ``yes'' or ``no'' answer to the question if a particular value can be proven to bound $\dmGH(X, Y)$ from below. In order to obtain an informative lower bound, one wants to find the largest value for which the answer is ``yes''. Since $\dmGH(X, Y) \leq \frac{1}{2}d_{\max}$, the answer must be ``no'' for any value above $\frac{1}{2}d_{\max}$, and therefore it suffices to limit the scope to $(0, \frac{1}{2}d_{\max}]$. To avoid redundancy when checking the values from this interval, we consider the following result.

\begin{claim}
	Let $\Delta$ denote the set of absolute differences between the distances in $X$ and $Y$, $$\Delta \defeq \{\left|d_X(x, x') - d_Y(y, y')\right|: x, x' \in X, y, y' \in Y\},$$ and let $\{\delta_i\}_{i=1}^{|\Delta|}$ represent the sorting order of $\Delta$, $0 = \delta_1 < \ldots < \delta_{|\Delta|} = d_{\max}$. If $\delta_i < d_1 < d_2 \leq \delta_{i+1}$ for some $d_1, d_2$ and $i \in \I{|\Delta|-1}$, then $\textproc{VerifyLowerBound}(\DX, \DY, d_1) = \textproc{VerifyLowerBound}(\DX, \DY, d_2)$.
	\label{claim 6}
\end{claim}
\begin{proof}
	\textproc{VerifyLowerBound} considers the value of its argument $d$ only through comparisons of the form ``$\delta < d$'' for some $\delta$, that occur in \textproc{FindLarge$K$} and \textproc{SolveFeasibleAssignment}. Notice that the values of $\delta$ compared with $d$ in \textproc{FindLarge$K$} are the entries of $\DX$ or $\DY$, and therefore belong to $\Delta$ as $$\{d_X(x, x'): x, x' \in X\}, \{d_Y(y, y'): y, y' \in Y\} \subseteq \Delta.$$ The values of $\delta$ compared with $d$ in \textproc{SolveFeasibleAssignment} belong to $\Delta$ by construction.
	
	For any $\delta \in \Delta$, $\delta < d_1$ if and only if $\delta < d_2$. This is because $\delta < d_2$ implies $\delta < d_1$ from $\delta \notin [d_1, d_2)$, while $\delta < d_1$ implies $\delta < d_2$ trivially. It follows that both \textproc{FindLarge$K$} and \textproc{SolveFeasibleAssignment} yield identical outputs on $d_1$ and $d_2$ (and otherwise identical inputs), and hence so does \textproc{VerifyLowerBound}.
\end{proof}

Claim \ref{claim 6} implies that the largest $\delta \in \Delta$ s.t. $\textproc{VerifyLowerBound}(\DX, \DY, \delta) = \mathrm{TRUE}$ is the largest $d \in \R$ s.t. $\textproc{VerifyLowerBound}(\DX, \DY, d) = \mathrm{TRUE}$. We use this fact to implement the procedure \textproc{FindLowerBound} (see Appendix \ref{FindLowerBound} for the pseudocode), that obtains a lower bound of $\dmGH(X, Y)$ by calling $\textproc{VerifyLowerBound}(\DX, \DY, \delta)$ for each $\delta \in \Delta$ from largest to smallest, and stops once the output is TRUE. Since $|\Delta| = O(N^4)$ in the general case, the time complexity of \textproc{FindLowerBound} is $O(N^7)$. 

\begin{remark}
	Using binary search on $\Delta$ instead of traversing its values in descending order reduces the number of calls to \textproc{VerifyLowerBound} from $O(N^4)$ to $O(\log N)$, bringing the time complexity of \textproc{FindLowerBound} to $O(N^3 \log N)$. We however notice that, given some $d_1 < d_2$, $\textproc{VerifyLowerBound}(\DX, \DY, d_2) = \mathrm{TRUE}$ does not guarantee $\textproc{VerifyLowerBound}(\DX, \DY, d_1) = \mathrm{TRUE}$ due to the heuristic nature of \textproc{FindLarge$K$} (for example, it is possible to obtain $\widetilde{M}(X, d_1) \leq |Y| < \widetilde{M}(X, d_2)$ even though, trivially, $M(X, d_1) \geq M(X, d_2) \geq \widetilde{M}(X, d_2)$). It follows that relying on the binary search in \textproc{FindLowerBound} can result in failing to find the largest $\delta \in \Delta$ s.t. $\textproc{VerifyLowerBound}(\DX, \DY, \delta) = \mathrm{TRUE}$, and thus in reducing time complexity at the cost of potentially lower accuracy.
\end{remark}
	
\begin{remark}
    An increase in accuracy at the expense of time complexity can be achieved by modifying \\$\textproc{FindLowerBound}(\DX, \DY)$ and its subroutines so that Theorem \ref{theorem 2} for the $d$-bounded distance samples tries to prove not only $\dmGH(X, Y) \geq \frac{d}{2}$ but also $\dmGH(X, Y) \geq \frac{d^*}{2}$ for some $d^* < d$, $d^* \in \Delta$. For every pair $i \in \I{n}, j \in \I{|Y|}$, the modified $\textproc{CheckTheoremB}(K, \DY, d)$ finds the largest $d^* \leq d$ s.t. $\textproc{SolveFeasibleAssignment}(\row{i}(K), \row{j}(\DY), d^*)=\mathrm{TRUE}$  (instead of checking the existence of feasible assignment for $d$ only) and returns the largest $d^*$ s.t. some $i$ satisfies $\|\row{i}(K) - \row{i}(L)\|_\infty \geq d^* \quad \forall L \in \ps_n(\DY)$. Accordingly, \textproc{VerifyLowerBound} propagates this value of $d^*$ to \textproc{FindLowerBound}, which stores their maximum from the already processed part of $\Delta$ (and proceeds only if this maximum is smaller than the current $\delta_i$). Because the existence of feasible assignment for $d^*$ is monotonic in $d^*$, using binary search on $\Delta$ is sufficient to find the largest $d^* \leq d$ s.t. $\textproc{SolveFeasibleAssignment}(\row{i}(K), \row{j}(\DY), d^*)=\mathrm{TRUE}$. It follows that the increase in time complexity of the modified \textproc{CheckTheoremB} (and therefore \textproc{FindLowerBound}) is only by a factor of $O(\log N)$.
    
    The modification would be able to provide better bounds by additionally trying smaller distance samples with larger entries in Theorem \ref{theorem 2}. In particular, it is guaranteed to prove the baseline lower bound $\dmGH(X, Y) \geq \frac{1}{2}|\diam X - \diam Y|$. Assuming without loss of generality $\diam X \geq \diam Y$, the call to $\textproc{FindLarge$K$}(\DX, \diam X)$ from $\textproc{VerifyLowerBound}(\DX, \DY, \diam X)$ will yield $K$ with non-zero entries that each are at least $(\diam X - \diam Y)$ away from all $\DY$ entries. Because the $(\diam X)$-bounded $K$ is also $(\diam X - \diam Y)$-bounded, the modified $\textproc{CheckTheoremB}(K, \DY, \diam X)$ will output $d^* \geq \diam X - \diam Y$, ensuring that the resulting lower bound is at least $\frac{1}{2}|\diam X - \diam Y|$.
    
    This is in contrast with the original \textproc{FindLowerBound}, in which the only easy guarantee for proving the baseline lower bound via Theorem \ref{theorem 2} is the presence of $\diam X$ values in $K$ produced by $\textproc{FindLarge$K$}(\DX,\\ \diam X - \diam Y)$. However, because $\textproc{FindLarge$K$}$ prioritizes the number of entries $\geq \diam X - \diam Y$ in a row of $\DX$ over their magnitude, and because a point with a distance $\diam X$ can have arbitrarily many distances $< \diam X - \diam Y$, establishing practical conditions for the output of \textproc{FindLowerBound} to be $\geq \frac{1}{2}(\diam X - \diam Y)$ is non-trivial.
\end{remark}

%% file: text/4_upper_bound.tex
\section{Upper bound for $\dmGH(X, Y)$}
\label{upper bound}
\subsection{Construction method}
To obtain an upper bound of $\dmGH(X, Y)$, we recall the definition $$\dmGH(X, Y) \defeq \frac{1}{2}\max \{\inf_{\varphi} \dis \varphi , \inf_{\psi} \dis \psi \},$$ where $\varphi$ and $\psi$ are the mappings $\varphi:X \to Y$ and $\psi:Y \to X$ and the infimums are taken over the corresponding function spaces. It follows that $\dmGH(X, Y) \leq \frac{1}{2}\max \{\dis \varphi, \dis \psi \}$ for any particular choice of $\varphi$ and $\psi$. Considering the exponential size of function spaces, we rely on a small randomized sample of mappings to tighten the upper bound. To sample $\varphi:X \to Y$, we use the construction method, a heuristic for solving quadratic assignment problems \cite{gilmore1962optimal, burkard1998quadratic}. The construction method iteratively maps each $x \in X$ to some $y \in Y$, chosen by a greedy algorithm to minimize $\dis \varphi$ as described below.

Let $X = \{x_1, \ldots, x_{|X|}\}$ and $Y = \{y_1, \ldots, y_{|Y|}\}$. We randomly choose a permutation $\pi$ of $\I{|X|}$ to represent the order in which the points in $X$ are mapped. At step $i$ we map $x_{\pi(i)}$ by choosing $y_{j_i}$ and setting $\varphi(x_{\pi(i)}) \defeq y_{j_i}$. We represent these choices by inductive construction $R_\varphi^{(i)} = R_\varphi^{(i-1)} \cup \{(x_{\pi(i)}, y_{j_i})\}$ for $i = 1, \ldots, |X|$, where $R_\varphi^{(0)} \defeq \emptyset$. The particular choice of $y_{j_i}$ at step $i$ is made to minimize distortion of resultant $R_\varphi^{(i)}$: $$y_{j_i} \in \argmin_{y \in Y}\,\dis\Big(R_\varphi^{(i-1)} \cup \big\{(x_{\pi(i)}, y)\big\}\Big).$$ After carrying out all $|X|$ steps, $\varphi:X \to Y$ is given by the constructed relation $R_\varphi \defeq R_\varphi^{(|X|)}$, and by definition, $\dis \varphi = \dis R_\varphi$.

Notice the possible ambiguity in the choice of $y_{j_i}$ when $y \in Y$ minimizing $\dis\Big(R_\varphi^{(i-1)} \cup \big\{(x_{\pi(i)}, y)\big\}\Big)$ is not unique. In particular, any $y \in Y$ can be chosen as $y_{j_1}$ at step 1, since $\dis \big\{(x_{\pi(i)}, y_{j_1})\big\} = 0$ is invariant to the said choice. In the case of such ambiguity, our implementation simply decides to map $x_{\pi(i)}$ to $y_{j_i}$ of the smallest index $j_i$. However, in applications one might want to modify this logic to leverage the knowledge of the relationship between the points from two metric spaces.

We formalize the above heuristic under a randomized, $O(N^3)$-time procedure \textproc{SampleSmallDistortion} (see Appendix \ref{SampleSmallDistortion} for the pseudocode) that samples a mapping between the two metric spaces with the intent of minimizing its distortion, and outputs this distortion. We then describe an algorithm \textproc{FindUpperBound} (see Appendix \ref{FindUpperBound} for the pseudocode), that repeatedly calls \textproc{SampleSmallDistortion} to find $\varphi^*:X \to Y$ and $\psi^*:Y \to X$, the mappings of the smallest distortion among those sampled from $X \to Y$ and $Y \to X$, respectively, and finds an upper bound for $\dmGH(X, Y)$ as $\frac{1}{2}\max \{\dis \varphi^*, \dis \psi^*\}$. The time complexity of \textproc{FindUpperBound} is therefore $O(sN^3)$, where $s$ is the total number of sampled mappings.

\subsection{Approximation ratio bounds}
Unlike the lower bound estimation, our heuristic for the upper bound operates within the solution space of $\dmGH$ formulation as a combinatorial minimization problem. This allows leveraging a result for bottleneck assignment problems bounding the ratio between the two extrema of the objective function (see e.g. Theorem 2 in \cite{burkard1985probabilistic}) to understand better the convergence of \textproc{SampleSmallDistortion}. Indeed, solving $\inf_{\phi:X\to Y} \dis \phi$ can be cast as a variant of the QBAP with the cost function $c(i, j, k, h) \defeq |\DX_{i,j} - \DY_{k,h}|$ that allows non-injective assignments and hence has the search space of size $|Y|^{|X|}$. Assuming $|X| > 1$, the maximum of its objective function $\dis \phi$ on the feasible region is $d_{\max}$. We consider a scaled but otherwise identical minimization problem with the cost function $\widetilde{c}(i, j, k, h) \defeq \frac{|\DX_{i,j} - \DY_{k,h}|}{d_{\max}} \in [0, 1]$, whose objective function has the minimum of $\frac{\inf_{\phi:X\to Y} \dis \phi}{d_{\max}}$.
The result by Burkard and Finke \cite{burkard1985probabilistic} treats the values of $\widetilde{c}$ as identically distributed random variables, and states that if $|X|$ and $|Y|$ satisfy \begin{align*}\label{condition} \mathbb{P}\left[\frac{|\DX_{i,j} - \DY_{k,h}|}{d_{\max}} < 1-\frac{\epsilon}{2}\right] \leq 1-q \tag{$\star$} \end{align*} for some $\epsilon, q \in [0, 1]$, then the approximation ratio of the solution by the maximum is probabilistically bounded by $$\mathbb{P}\left[\frac{d_{\max}}{\inf_{\phi:X\to Y} \dis \phi} < 1 + \epsilon\right] \geq 1 - \exp(-q|X|^2 + |X| \log |Y|).$$

Let $b$ be the smallest distortion discovered by \textproc{SampleSmallDistortion} after a number of tries, and consider some $r \geq 1$. If for $\epsilon \defeq \frac{rd_{\max}}{b}-1$ we can find $q$ that satisfies both (\ref{condition}) and $q > \frac{\log |Y|}{|X|}$, the above result allows us to bound the probability that the approximation ratio of $b$ is worse than $r$:
\begin{align*}
    \mathbb{P}\left[\frac{b}{\inf_{\phi:X\to Y} \dis \phi} > r\right] &= \mathbb{P}\left[\frac{d_{\max}}{\inf_{\phi:X\to Y} \dis \phi} > \frac{rd_{\max}}{b}\right] \\&= 1 - \mathbb{P}\left[\frac{d_{\max}}{\inf_{\phi:X\to Y} \dis \phi} < 1 + \epsilon\right] \\&\leq \exp(-q|X|^2 + |X| \log |Y|).
\end{align*}
An analogous probabilistic bound applies to the minimization of distortion of $\psi:Y \to X$. These bounds can be used to adaptively control the sample size $s$ in \textproc{FindUpperBound} for achieving some target accuracy with the desired level of confidence.

%% file: text/5_algorithm.tex
\section{Algorithm for estimating $\dmGH(X, Y)$}
\label{algorithm}
The algorithm for estimating the mGH distance between compact metric spaces $X$ and $Y$ (``the algorithm'') consists of the calls $\textproc{FindLowerBound}(\DX, \DY)$ and $\textproc{FindUpperBound}(\DX, \DY)$. Notice that $\dmGH(X, Y)$ is found exactly whenever the outputs of the two procedures match. Time complexity of the algorithm is $O(N^7)$ whenever $s$, the number of mappings sampled from $X \to Y$ and $Y \to X$ in \textproc{FindUpperBound}, is limited by $O(N^4)$. 


To obtain a more practical result, we now consider a special case of metric spaces induced by unweighted undirected graphs. We show that estimating the mGH distance between such metric spaces in many applications has time complexity $O(N^3 \log N)$, and present our implementation of the algorithm in Python.

\subsection{$\dmGH$ between unweighted undirected graphs}
Let $G = (V_G, E_G)$ be a finite undirected graph. For every pair of vertices $v, v' \in V_G$, we define $d_G(v, v')$ as the shortest path length between $v$ to $v'$. If weights of the edges in $E_G$ are positive, the resulting function $d_G : V_G \times V_G \to [0, \infty]$ is a metric on $V_G$. We say that the metric space $(V_G, d_G)$ is induced by graph $G$, and notice that its size $|V_G|$ is the order of $G$. By convention, the shortest path length between vertices from different connected components of a graph is defined as $\infty$, and therefore $(V_G, d_G)$ is compact if and only if graph $G$ is connected.

For brevity, we will use notation $G \defeq (V_G, d_G)$, assuming that the distinction between graph $G$ and metric space $G$ induced by this graph is clear from the context. In particular, we refer to the mGH distance between compact metric spaces induced by undirected connected graphs $G$ and $H$ as $\dmGH(G, H)$, or even call it ``the mGH distance between graphs $G$ and $H$''.

Let $G, H$ be connected unweighted undirected graphs. Notice that $\dmGH(G, H) = 0$ if and only if graphs $G$ and $H$ are isomorphic. We use the following result to reduce the computational cost of estimating $\dmGH(G, H)$ as compared to that in the general case.

\begin{claim}
	If $G$ is a finite connected unweighted undirected graph, all entries in distance matrix $D^G$ of the corresponding compact metric space are from $\{0, 1, \ldots, \diam G\}$.
	\label{claim 7}
\end{claim}
Claim \ref{claim 7} implies that there are at most $d_{\max} + 1$ distinct entries in any distance sample of either $G$ or $H$, where $d_{\max} \defeq \max \{\diam G, \diam H\}$. Recall that \textproc{SolveFeasibleAssignment} traverses the two sorted vectors strictly left to right, in each vector keeping track only of its part after the last assigned entry. It follows that both vectors can be represented simply as counts of entries $0, 1, \ldots, d_{\max}$ they contain, with traversing a vector corresponding to decrementing the leftmost non-zero count. In particular, assigning an entry in one vector to an entry in another is then equivalent to decrementing the corresponding counts in each of the vectors, which reflects the injectivity of the assignment being constructed. Notice that if $v_i$ is assigned to $u_j$ at some step and the remaining entries of the two vectors contain $v_i$ and $u_j$, respectively, \textproc{SolveFeasibleAssignment} will make an identical assignment at the next step. It follows that representing the vectors by counting their identical entries allows the procedure to make bulk assignments, which reduces the time complexity of \textproc{SolveFeasibleAssignment} from $O(N)$ to $O(d_{\max})$ and makes the complexity of \textproc{CheckTheoremB} $O(N^2 d_{\max})$, where $N \defeq \max \{|V_G|, |V_H|\}$.

\begin{remark}
	From the perspective of optimization theory, representing vectors as counts of their entries reformulates the linear assignment feasibility problem of \textproc{SolveFeasibleAssignment} as a transportation feasibility problem.
\end{remark}

Another implication of Claim \ref{claim 7} narrows the set of absolute differences between the distances in $G$ and $H$ to $\Delta = \{0, 1, \ldots, d_{\max}\}$, reducing the time complexity of traversing its elements from $O(N^4)$ to $O(d_{\max})$. This bounds the time complexity of \textproc{FindLowerBound} by $O(N^3 d_{\max})$. The complexity of the entire algorithm is therefore $O(N^3 d_{\max})$ when the number of sampled mappings is $s = O(d_{\max})$.

Network diameter often scales logarithmically with network size, e.g. in Erd\H{o}s--R\'enyi random graph model \cite{albert2002statistical} and Watts--Strogatz small-world model \cite{watts1998collective}, or even sublogarithmically, e.g. in the configuration model with i.i.d. degrees \cite{hofstad2007phase} and Barab\'asi--Albert preferential attachment model \cite{cohen2003scale}. This suggests the time complexity of the algorithm in applications to be $O(N^3 \log N)$, deeming it practical for shape comparison for graphs of up to a moderate order.

\subsection{Implementation}
\label{implementation}
We have implemented the algorithm for estimating mGH between unweighted graphs in Python 3.7 as part of \verb|scikit-tda| package \cite{saul2019scikit} (\url{https://github.com/scikit-tda}). Our implementation takes adjacency matrices (optionally in sparse format) of unweighted undirected graphs as inputs. If an adjacency matrix corresponds to a disconnected graph, the algorithm approximates it with its largest connected component. The number of mappings to sample from $X \to Y$ and $Y \to X$ in \textproc{FindUpperBound} is parametrized as a function of $|X|$ and $|Y|$.



%% file: text/6_computational_examples.tex
\section{Computational examples}
\label{computational examples}
This section demonstrates the performance of the algorithm on some real-world and synthetic networks. The real-world networks were sourced from the Enron email corpus \cite{klimt2004enron}, the cybersecurity dataset collected at Los Alamos National Laboratory \cite{kent2016cyber}, and the functional magnetic resonance imaging dataset ABIDE I \cite{di2014autism}.

\subsection{Methodology and tools}
\label{methodology and tools}
To estimate mGH distances between the graphs, we use our implementation of the algorithm from subsection \ref{implementation}. We set the number of mappings from $X \to Y$ and $Y \to X$ to sample by the procedure \textproc{FindUpperBound} to $\left\lceil \sqrt{|X|} \log(|X|+1) \right\rceil$ and $\left\lceil \sqrt{|Y|} \log(|Y|+1) \right\rceil$, respectively.

Given a pair of connected graphs, we approximate the mGH distance between them by $\tilde{d} \defeq \frac{b_L + b_U}{2}$, where $b_L$ and $b_U$ are the lower and upper bounds produced by the algorithm. The relative error of the algorithm is estimated by $$\eta \defeq \begin{cases}
\frac{b_U - b_L}{2\tilde{d}} = \frac{b_U - b_L}{b_L + b_U}, & \text{if $b_U > 0$}\\
0, & \text{if $b_U = 0$}
\end{cases},$$ noting that the case of $b_U = 0$ implies that the algorithm has found the mGH distance of 0 exactly. In addition, we compute the utility coefficient defined by $$\upsilon \defeq \begin{cases}
\frac{b_L - b'_L}{2\tilde{d}} = \frac{b_L - b'_L}{b_L + b_U}, & \text{if $b_U > 0$}\\
0, & \text{if $b_U = 0$}
\end{cases},$$ where $b'_L \leq b_L$ is the baseline lower bound: $b'_L \defeq \frac{1}{2}|\diam X - \diam Y| \leq \dmGH(X, Y)$ for $X, Y \in \MM$. The utility coefficient thus quantifies the tightening of the lower bound achieved by using Theorems \ref{theorem 1} and \ref{theorem 2}. 

Due to the possible suboptimality of the mappings selected by using the construction method (see section \ref{upper bound}), the upper bound may not be computed accurately enough. From the definition of relative error $\eta$ and utility coefficient $\upsilon$, a sufficiently loose upper bound can make $\eta$ arbitrarily close to 1 and $\upsilon$ --- arbitrarily small.

We measured $\eta$ and $\upsilon$ separately for each dataset. For the real-world data, we also used the approximated distances $\tilde{d}$ to identify graphs of outlying shapes and matched these graphs to events or features of importance in application domains, following the approach taken in e.g. \cite{bunke2004classification, bunke2006computer, koutra2013deltacon, pincombe2005anomaly}. Unlike \cite{bunke2004classification, koutra2013deltacon}, and \cite{pincombe2005anomaly} that focus on local time outliers (under the assumption of similarity between graphs from consecutive time steps), we considered the outliers with respect to the entire time range (where applicable), similarly to \cite{bunke2006computer}. 

To identify the outliers, we applied the Local Outlier Probability (LoOP) method \cite{kriegel2009loop} to the graphs using their approximated pairwise mGH distances. LoOP uses a local space approach to outlier detection and is robust with respect to the choice of parameters \cite{kriegel2009loop}. The score assigned by LoOP to a data object is interpreted as the probability of the object to be an outlier. We binarize the score into outlier/non-outlier classification using the commonly accepted confidence levels of 95\%, 99\%, and 99.9\% as thresholds, chosen empirically for each individual dataset. We used the implementation of LoOP in Python \cite{constantinou2018pynomaly} (version 0.2.1), modified by us to allow non-Euclidean distances between the objects. We ran LoOP with locality and significance parameters set to $k = 20$ and $\lambda = 1$, respectively.

The synthetic graphs were generated according to Erd\H{o}s--R\'enyi, Watts--Strogatz, and Barab\'asi--Albert network models. We used implementations of the models provided in \cite{hagberg2008exploring} (version 2.1).

All computations were performed on a single 2.70GHz core of Intel i7-7500U CPU.

\subsection{Enron email corpus}
Enron email corpus (available at \url{https://www.cs.cmu.edu/~./enron/}) represents a collection of email conversations between the employees, mostly senior management, of the Enron corporation from October 1998 to March 2002. We used the latest version of the dataset from May 7, 2015, which contains roughly 500K emails from 150 employees.

Associating employees with graph vertices, we view the dataset as a dynamic network whose 174 instances reflect weekly corporate email exchange over the course of 3.5 years. An (unweighted) edge connecting a pair of vertices in a network instance means a mutual exchange of at least one email between the two employees on a particular week. The time resolution of 1 week was suggested in \cite{sulo2010meaningful} for providing an appropriate balance between noise reduction and information loss in Enron dataset.

We expected major events related to the Enron scandal in the end of 2001 to cause abnormal patterns of weekly email exchange between the senior management, distorting the shape of the corresponding network instances. As a consequence, metric spaces generated by such network instances would be anomalously far from the rest with respect to the mGH distance.

\begin{figure*}[!b]
	\includegraphics[width=\linewidth]{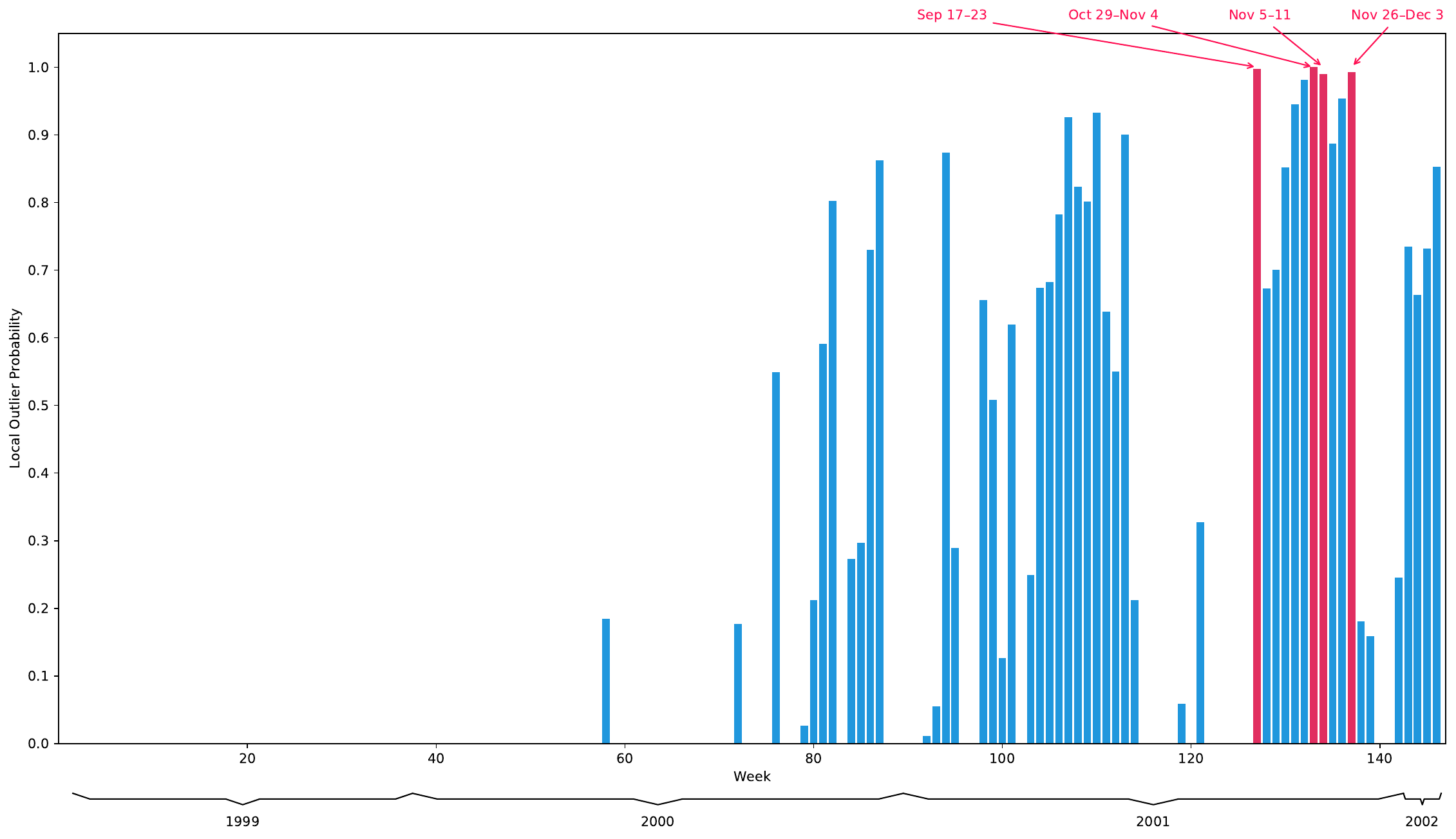}
	\caption{Outlier probabilities assigned to the weekly email exchange networks. Red indicates outlier probabilities $> 0.99$, corresponding to the weeks of Sep 17, Oct 29, Nov 5, and Nov 26 in the year 2001.}
	\label{fig:enron_loop}
\end{figure*}

In preparation for the analysis, we discarded all empty network instances corresponding to the weeks of no email exchange between the employees (of which all 28 weeks happened before May 1999). Each of the remaining 146 graphs was then replaced with its largest connected component. The distribution of the order of the resulting graphs had a mean of 68.2, a standard deviation of 99.8, and a maximum of 706.

We estimated the mGH distances in all 10,585 distinct pairs of the non-empty connected network instances. Average graph order and computing time per one pair were distributed as $68.2 \pm 70.3$ and $0.93\mathrm{s} \pm 3.91\mathrm{s}$, respectively (where $\mu \pm \sigma$ refers to distribution with a mean of $\mu$ and a standard deviation of $\sigma$; no assumptions of normality are made, and we use standard deviation solely as a measure of spread). The algorithm found exact mGH distances in 74.4\% of the graph pairs, with relative error $\eta$ and utility coefficient $\upsilon$ distributed as $0.057 \pm 0.118$ and $0.043 \pm 0.085$, respectively. The ratio between the means of $\upsilon$ and $\eta$ (i.e. $\frac{\upsilon}{\eta} + 1 = \frac{\upsilon + \eta}{\eta} = \frac{b_U - b_{L'}}{b_L - b_{L'}}$) implies that using Theorems \ref{theorem 1} and \ref{theorem 2} on average reduced the relative error by a factor of 1.75.

We ran LoOP on the network instances using their (approximated) pairwise mGH distances $\tilde{d}$. The resulting outlier probability assigned to each network instance (Figure \ref{fig:enron_loop}) thus measures the abnormality of its shape.

To see if the abnormal shape of email exchange corresponds to events of high importance from the Enron timeline, we applied the threshold of 0.99 to the outlier probabilities. Three out of four network instances that scored above the threshold correspond to the weeks of known important events in 2001, namely the weeks of Oct 29, Nov 5, and Nov 26 (each date is a Monday). As the closing stock price of Enron hit an all-time low on Friday, Oct 26, Enron's chairman and CEO Kenneth Lay was making multiple calls for help to Treasure Secretary Paul O'Neill and Commerce Secretary Donald Evans on Oct 28--29. Enron fired both its treasurer and in-house attorney on Nov 5, admitted to overstating its profits for the last five years by \$600M on Nov 8, and agreed to be acquired by Dynegy Inc. for \$9B on Nov 9. On Nov 28, Dynegy Inc. aborted the plan to buy Enron, and on Dec 2, Enron went bankrupt.

We conclude that the abnormal shape of email exchange networks tends to correspond to disturbances in their environment, and that the algorithm estimates the mGH distance accurately enough to capture it.

\subsection{LANL cybersecurity dataset}
Los Alamos National Laboratory (LANL) cybersecurity dataset (available at \url{https://csr.lanl.gov/data/cyber1/}) represents 58 consecutive days of event data collected from LANL's corporate computer network \cite{kent2016cyber}. For our purposes, we considered its part containing records of authentication events, generated by roughly 11K users on 18K computers, and collected from individual Windows-based desktops and servers. During the 58-day data collection period, a “red team” penetration testing operation had taken place. As a consequence, a small subset of authentications were labeled as red team compromise events, presenting well-defined bad behavior that differed from normal user and computer activity. The labeling is not guaranteed to be exhaustive, and authentication events corresponding to red team actions, but not labeled as such, are likely to be present in the data. \cite{heard2016network}

Each authentication event occurs between a pair of source and destination computers. Viewing the computers as graph vertices, we associated each user with a dynamic network, whose instances reflect their daily authentication activity within the 58-day period. An (unweighted) edge connecting a pair of vertices in a network instance means that at least one authentication event by the user has occurred between the two computers on a particular day. The user-based approach to graph representation of the data aims to capture the patterns of user account misuse that are expected to occur during a cyberattack.

Our objective was to develop an unsupervised approach that can identify the red team activity associated with a user's account. We expected that frequent compromise events within the course of one day should distort the shape of the corresponding network instance. As a consequence, metric spaces generated by such network instances would be anomalously far from the rest.

For the analysis, we selected 20 users with the highest total of associated red team events and randomly chose another 20 users from those unaffected by the red team activity (see Figure \ref{fig:redteam}). Each of their 40 dynamic networks initially comprised of 58 instances. We discarded all empty network instances corresponding to the days of inactivity of a user account, and replaced each of the remaining 1,997 graphs with its largest connected component. The distribution of the order in the resulting graphs had a mean of 32.7 and a standard deviation of 75.8. The largest graphs were associated with the red team-affected user U1653@DOM1 --- their order was distributed with a mean of 178.2, a standard deviation of 391.5, and a maximum of 2,343.

\begin{figure*}[!t]
	\includegraphics[width=\linewidth]{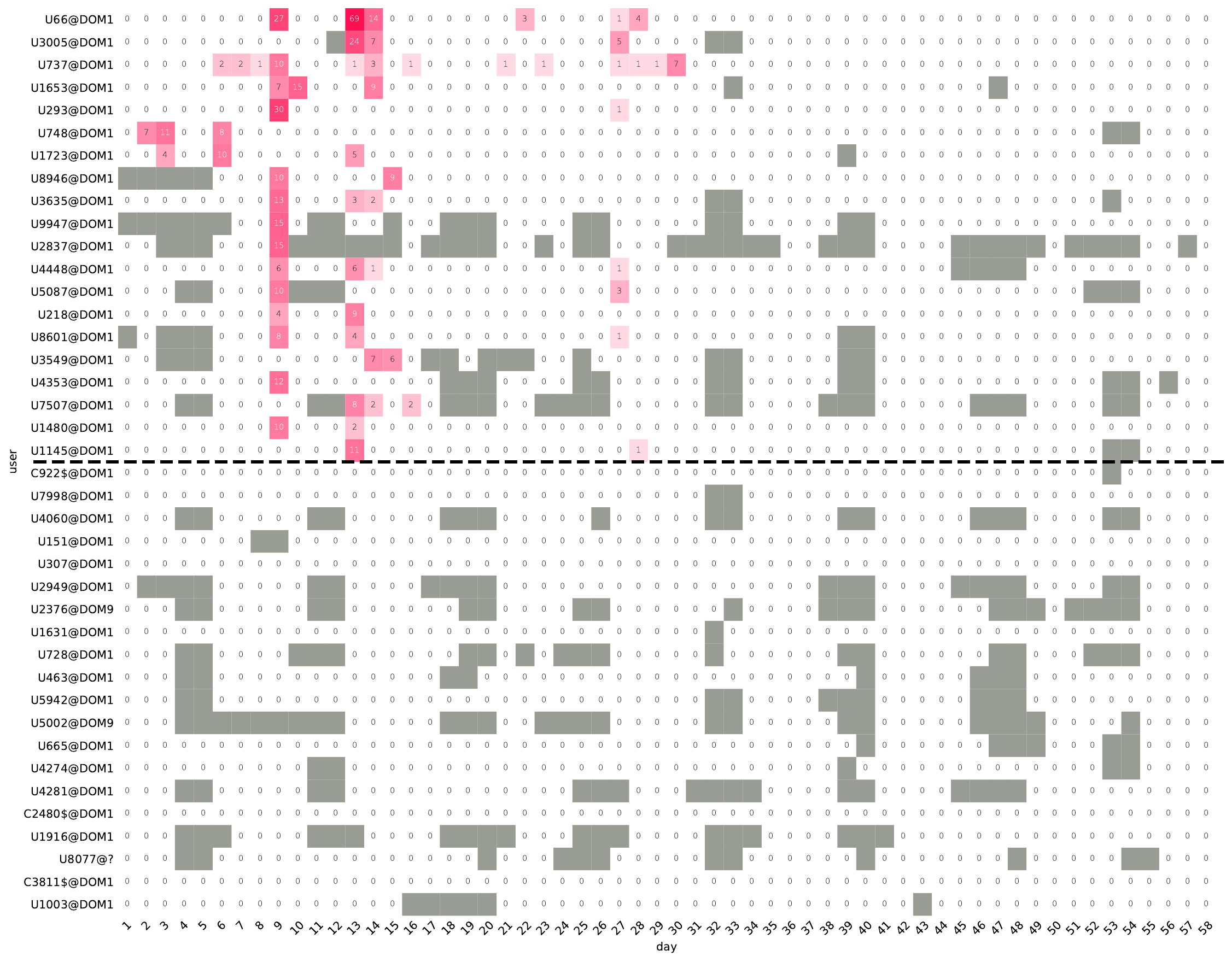}
	\caption{Frequency of red team events in daily authentication activity of the selected users. Grey indicates days of no authentication activity by user. Dashed line separates the two groups of 20 users.}
	\label{fig:redteam}
\end{figure*}

Separately for each of the selected users, we estimated the mGH distances in all distinct pairs of the non-empty connected network instances associated with her account. Table \ref{tab:LANL performance} shows average graph order in a pair and the performance metrics of the algorithm, aggregated per user subsets. We notice that using Theorems \ref{theorem 1} and \ref{theorem 2} has reduced the relative error by a factor of 3.5 on average. In addition, the algorithm did not seem to perform worse on the larger graphs associated with user U1653@DOM1.

\begin{table*}[!t]
	\begin{tabular}{|c|c|c|c|c|c|c|}
		\cline{2-7}
		\multicolumn{1}{c|}{} & \# of & average & computing & exact & relative & utility \\
		\multicolumn{1}{c|}{} & pairs & graph order & time & distances & error $\eta$& coefficient $\upsilon$ \\
		\hhline{|-|-|-|-|-|-|-|}
		\multirow{2}{*}{all 40 users} & \multirow{2}{*}{50316} & \multirow{2}{*}{$34.2 \pm 60.7$} & \multirow{2}{*}{$0.44 \mathrm{s} \pm 14.09 \mathrm{s}$} & \multirow{2}{*}{84.7\%}&  \multirow{2}{*}{$0.049 \pm 0.121$}&  \multirow{2}{*}{$0.121 \pm 0.127$}\\
		&&&&&&\\
		\hline
		\multirow{2}{*}{U1653@DOM1} & \multirow{2}{*}{1540} & \multirow{2}{*}{$178.2 \pm 274.3$} & \multirow{2}{*}{$13.52 \mathrm{s} \pm 79.46 \mathrm{s}$} & \multirow{2}{*}{94.2\%}&  \multirow{2}{*}{$0.012 \pm 0.051$}&  \multirow{2}{*}{$0.135 \pm 0.123$}\\
		&&&&&&\\
		\hline
		\multirow{2}{*}{other 39 users} & \multirow{2}{*}{48776} & \multirow{2}{*}{$29.7 \pm 27.3$} & \multirow{2}{*}{$0.028 \mathrm{s} \pm 0.058 \mathrm{s}$} & \multirow{2}{*}{84.4\%}&  \multirow{2}{*}{$0.051 \pm 0.122$}&  \multirow{2}{*}{$0.120 \pm 0.127$}\\
		&&&&&&\\
		\hline
	\end{tabular}
	\vspace*{1em}
	\caption{Performance of the algorithm on user-based daily authentication graphs. $\mu \pm \sigma$ denotes that the distribution of a variable across the graph pairs has a mean of $\mu$ and a standard deviation of $\sigma$.}
	\label{tab:LANL performance}
\end{table*}

\begin{figure*}[!t]
	\includegraphics[width=\linewidth]{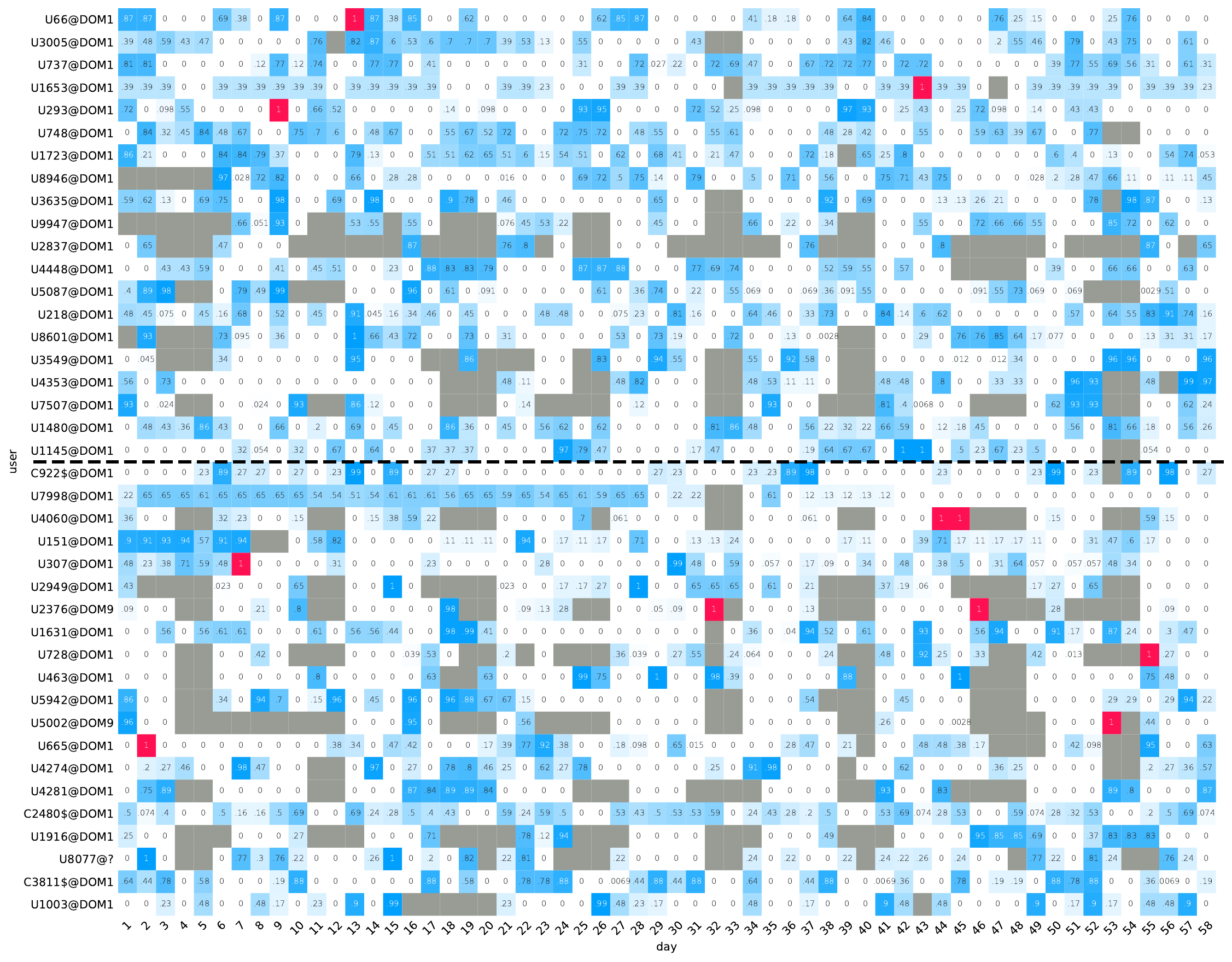}
	\caption{Outlier probability assigned to user-based daily authentication graphs. Red indicates outlier probabilities $> 0.999$. Grey indicates empty graphs (excluded from analysis). The dashed line separates the two groups of 20 users.}
	\label{fig:loop}
\end{figure*}

Separately for each user, we ran LoOP on the associated network instances using their (approximated) pairwise mGH distances $\tilde{d}$. The resulting outlier probability assigned to each network instance (Figure \ref{fig:loop}) thus measures the abnormality of its shape for the particular user.

To see if the days of high compromise activity can be identified from the abnormal shape of the corresponding network instances, we approached the identification as a binary classification task. User's daily activity is considered a compromise if it includes at least 30 red team events, and is predicted as such if the outlier probability assigned to the corresponding network instance is $>0.999$. The resulting confusion matrix of our shape-based binary classifier is shown in Table \ref{confusion matrix}, establishing its accuracy, precision, and recall as 99.5\%, 18.2\%, and 100\%, respectively.

\begin{table*}[!t]
	\begin{tabular}{l|l|c|c|c}
		\multicolumn{2}{c}{}&\multicolumn{2}{c}{predicted}&\\
		\cline{3-4}
		\multicolumn{2}{c|}{}&compromise&not compromise&\multicolumn{1}{c}{total}\\
		\cline{2-4}
		\multirow{2}{*}{actual}& compromise & 2 & 0 & 2\\
		\cline{2-4}
		& not compromise & 9 & 1986 & 1995\\
		\cline{2-4}
		\multicolumn{1}{c}{} & \multicolumn{1}{c}{total} & \multicolumn{1}{c}{11} & \multicolumn{1}{c}{1986} & \multicolumn{1}{c}{1997}\\
	\end{tabular}
	\vspace*{.5em}
	\caption{Confusion matrix of the shape-based binary classifier.}
	\label{confusion matrix}
\end{table*}

Even though recall is prioritized over precision when identifying intrusions, low precision can be impractical when using the classifier alone. However, high recall suggests that combining the shape-based classifier with another method can improve performance of the latter. For example, classifying daily activity as a compromise if and only if both methods agree on it is likely to increase the other method's precision without sacrificing its recall.

We conclude that sufficiently frequent compromise behavior tends to distort the shape of networks representing user authentication activity, and that the algorithm estimates the mGH distance accurately enough to pick up the distortion. However, regular operations can cause similar distortion, and the shape of user authentication alone may be insufficient to accurately identify compromise behavior.

\subsection{ABIDE I dataset}
The Autism Brain Imaging Data Exchange I (ABIDE I, \url{http://fcon_1000.projects.nitrc.org/indi/abide/abide_I.html}) \cite{di2014autism} is a resting state functional magnetic resonance imaging dataset, collected to improve understanding of the neural bases of autism. Besides the diagnostic group (autism or healthy control), the information collected from the study subjects also includes their sex, age, handedness category, IQ, and current medication status. We considered the preprocessed version of ABIDE I used in \cite{kazeminejad2019topological}, containing brain networks of 816 subjects, including 374 individuals with autism spectrum disorder and 442 healthy controls.

Brain network of an individual is comprised of 116 nodes, each corresponding to a region of interest (ROI) in the automatic anatomical labeling atlas of the brain \cite{tzourio2002automated}. Connectivity of the network represents correlations between the brain activity in the ROI pairs, with the brain activity extracted as a time series for each region. Namely, an (unweighted) edge connects two nodes if Spearman's rank correlation coefficient between the corresponding pair of time series is in the highest 20\% of all such coefficients for distinct ROI pairs. This approach of thresholding connectivity to obtain an unweighted graph representation is commonly used for functional brain networks \cite{achard2007efficiency, supekar2008network, redcay2013intrinsic, rudie2013altered, rubinov2010complex}.

\begin{figure*}[!b]
	\includegraphics[width=\linewidth]{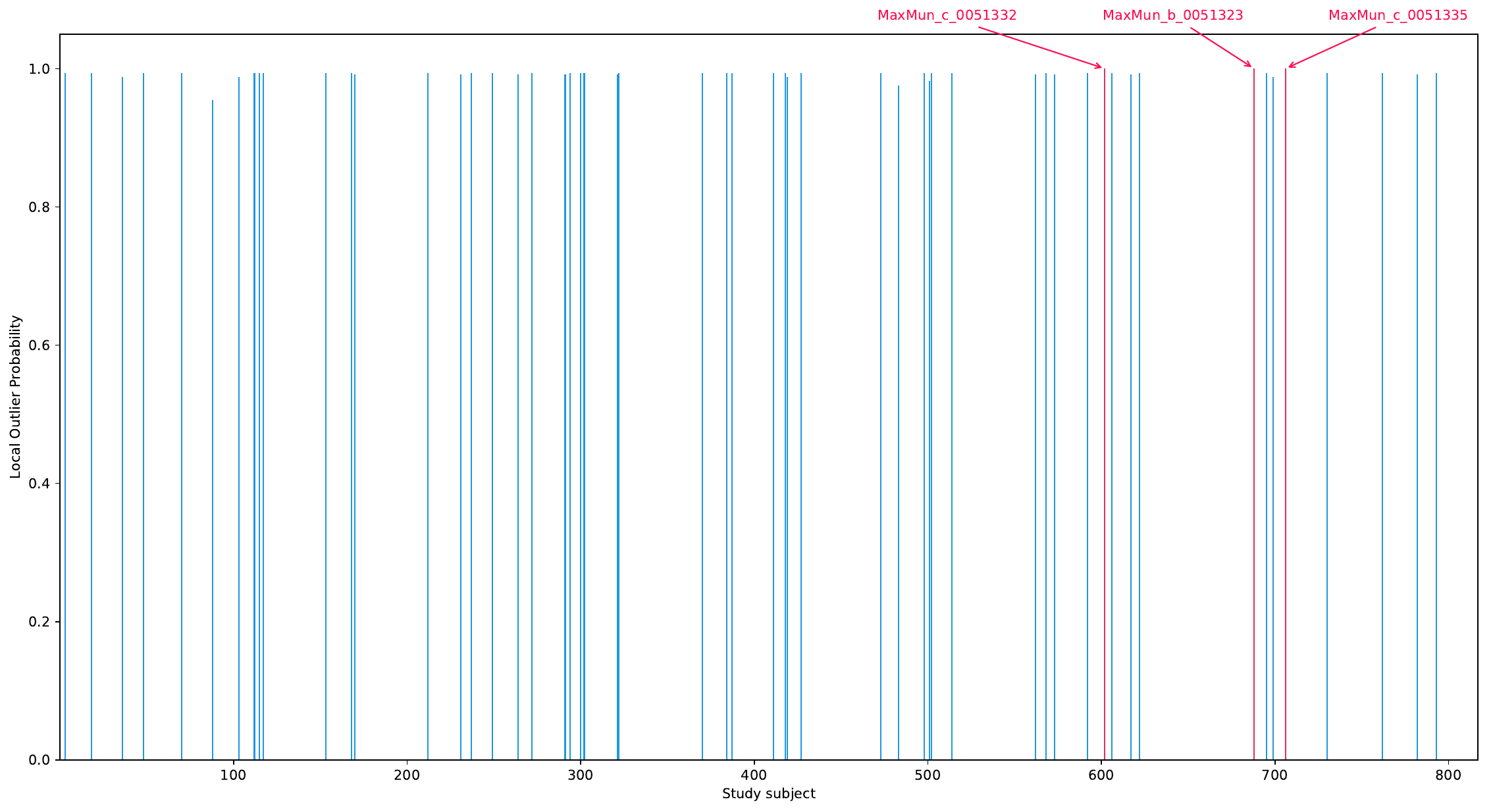}
	\caption{Outlier probability assigned to brain networks of study subjects. Blue indicates outlier probabilities $> 0.95$, and red --- outlier probabilities $> 0.999$. The latter correspond to the subjects MaxMun\_c\_0051332, MaxMun\_b\_0051323, and MaxMun\_c\_0051335. The remaining outlier probabilities are 0.}
	\label{fig:brain_loop}
\end{figure*}

We replaced the brain network of each subject with its largest component, which did not introduce a significant change to the data: the distribution of the order in the resulting 816 graphs had a mean of about 116.0 and a standard deviation of 0.1.

We estimated the mGH distances in all 332,520 distinct pairs of the connected brain networks. Average graph order and computing time per one pair were distributed as $116.0 \pm 0.1$ and $0.40\mathrm{s} \pm 20.29\mathrm{s}$, respectively. The algorithm found exact mGH distances in 78.6\% of the graph pairs, with relative error $\eta$ and utility coefficient $\upsilon$ distributed as $0.072 \pm 0.139$ and $0.361 \pm 0.214$, respectively. We notice that using Theorems \ref{theorem 1} and \ref{theorem 2} has reduced the relative error by a factor of 5 on average.

We ran LoOP on the brain networks using their (approximated) pairwise mGH distances $\tilde{d}$. The resulting outlier probability assigned to each brain network (Figure \ref{fig:brain_loop}) thus measures the abnormality of its shape.

To see if abnormal shape of a brain network corresponds to certain features of the individual, we applied the thresholds of 0.999 and 0.95 to the outlier probabilities. We were unable to identify the corresponding brain network shapes with outlying values of the available measurements, neither for the 3 subjects who scored above the threshold of 0.999 nor for the 54 subjects who scored above 0.95.

To see if the brain networks of subjects within the same diagnostic group tend to have similar shape, we performed cluster analysis based on the mGH distances $\tilde{d}$. We used \verb|scipy| implementation of hierarchical agglomerative clustering algorithm to split the 816 networks into two clusters (by the number of diagnostic groups in the dataset). The smaller cluster was comprised of the same 3 subjects who scored above the outlier probability threshold of 0.999. Discarding them as outliers and rerunning the analysis resulted in two clusters of 51 and 762 subjects, respectively. The clusters did not show any correspondence with the diagnostic groups, thus providing no evidence that the within-group mGH distances are smaller than the inter-group ones. However, we notice a significant overlap between the 54 subjects with an outlier probability above 0.95 and the cluster of 51 individuals, with 47 people shared between the two groups. This implies that the abnormal brain networks tend to be closer to one another than to the ``regular'' brain networks. This observation suggests that abnormality of the brain network shape is influenced by currently unknown features which are not included in the dataset.

We conclude that the algorithm estimates the mGH distance between Spearman's correlation-based functional brain networks with high accuracy. However, detected shape abnormalities do not seem to correspond to a conclusive pattern related to autism spectrum disorder identification.

\subsection{Synthetic networks}

To test performance of the algorithm on synthetic networks, we generated 100 graphs per each of Erd\H{o}s--R\'enyi (random), Watts--Strogatz (small-world), and Barab\'asi--Albert (scale-free) network models. The order $n$ of each graph was selected uniformly at random between 10 and 200, and other parameters of the model were based on $n$. In the Erd\H{o}s--R\'enyi $G(n, p)$ model, the probability $p$ for an edge between a vertex pair to appear was selected uniformly at random between $\frac{0.5\log n}{n}$ and $\frac{1.5\log n}{n}$. In the Watts--Strogatz $G(n, 2k, p)$ model, $k$, half of the average vertex degree, was selected uniformly at random between $1$ and $\lfloor 0.5 \log^2 n \rceil$, and the probability $p$ for an edge to get rewired was selected uniformly at random between $\frac{0.5\log n}{n}$ and $\frac{1.5\log n}{n}$. In the Barab\'asi--Albert $G(n, m)$ model, the number of edges $m$ to attach from a new node to the existing nodes was selected uniformly at random between $1$ and $\lfloor \log^2 n \rceil$. 

After generating the graphs, we replaced each of them with its largest connected component. For each set, we estimated the mGH distances in all distinct pairs of the 100 connected graphs therein. Table \ref{tab:synth performance} shows the average graph order in a pair and the performance metrics of the algorithm, aggregated per individual data sets.

\begin{table}[!t]
	\begin{tabular}{|c|c|c|c|c|c|c|}
		\cline{2-7}
		\multicolumn{1}{c|}{} & \# of & average & computing & exact & relative & utility \\
		\multicolumn{1}{c|}{} & pairs & graph order & time & distances & error $\eta$& coefficient $\upsilon$ \\
		\hhline{|-|-|-|-|-|-|-|}
		\multirow{2}{*}{Erd\H{o}s--R\'enyi} &\multirow{2}{*}{4950}& \multirow{2}{*}{$101.8 \pm 56.1$} & \multirow{2}{*}{$0.58 \mathrm{s} \pm 0.49 \mathrm{s}$} & \multirow{2}{*}{20.1\%}&  \multirow{2}{*}{$0.222 \pm 0.172$}&  \multirow{2}{*}{$0.054 \pm 0.089$}\\
		&&&&&&\\
		\hline
		\multirow{2}{*}{Watts--Strogatz} &\multirow{2}{*}{4950}& \multirow{2}{*}{$101.5 \pm 55.8$} & \multirow{2}{*}{$0.52 \mathrm{s} \pm 0.57 \mathrm{s}$} & \multirow{2}{*}{50.42\%}&  \multirow{2}{*}{$0.138 \pm 0.170$}&  \multirow{2}{*}{$0.004 \pm 0.028$}\\
		&&&&&&\\
		\hline
		\multirow{2}{*}{Barab\'asi--Albert} &\multirow{2}{*}{4950}& \multirow{2}{*}{$103.5 \pm 37.6$} & \multirow{2}{*}{$0.33 \mathrm{s} \pm 0.34 \mathrm{s}$} & \multirow{2}{*}{57.1\%}&  \multirow{2}{*}{$0.131 \pm 0.157$}&  \multirow{2}{*}{$0.035 \pm 0.088$}\\
		&&&&&&\\
		\hline
	\end{tabular}
	\vspace*{.5em}
	\caption{Performance of the algorithm on the synthesized networks. $\mu \pm \sigma$ denotes that distribution of a variable across the graph pairs has a mean of $\mu$ and a standard deviation of $\sigma$.}
	\label{tab:synth performance}
\end{table}

We notice that the algorithm performs significantly worse on the Erd\H{o}s--R\'enyi graphs. One possible explanation is that there are fewer identically connected vertices in random graphs than in those resembling real-world networks, which contributes to the combinatorial complexity of the search for distortion-minimizing mappings to obtain the upper bound. Recall from subsection \ref{methodology and tools} that inaccurate computations of the upper bound alone can have a detrimental effect on both $\eta$ and $\upsilon$.

Another interesting observation is that Theorems \ref{theorem 1} and \ref{theorem 2} have smaller utility when applied to the Watts--Strogatz graphs. Recall that a Watts--Strogatz small-world graph is generated from a lattice ring with each node connected to its $2k$ neighbors ($k$ for each side), by randomly rewiring a fraction (roughly, $p$) of its edges. For a small $p$, the rewired edges serve as shortcuts between the otherwise remote vertices and have a highly nonlinear effect on the diameter \cite{watts1998collective}. This allows for high variability in the diameters of generated graphs, thus contributing to the tightness of the baseline lower bounds $b'_L \defeq \frac{1}{2}|\diam X - \diam Y|$.

We conclude that the algorithm performs better on graphs with scale-free and small-world properties, observed in many real-world networks.

%% file: text/7_conclusion.tex
\section{Conclusion}
\label{conclusion}
The main contribution of this work is a feasible method for finding a lower bound for the mGH (and therefore GH) distance between finite metric spaces. The approach, based on the introduced notion of $d$-bounded distance samples, yields a polynomial-time algorithm for estimating the mGH distance. The algorithm is implemented as part of Python \verb|scikit-tda| library for the case of compact metric spaces induced by unweighted graphs. It is also shown that in the case of unweighted graphs of order $N$ whose diameter scales at most logarithmically with $N$ the algorithm has time complexity $O(N^3 \log N)$. 

To test the algorithm's performance, we applied it to both real and synthesized networks. Among the synthesized networks we tested, the best performance was observed for graphs with scale-free and small-world properties. We have also found that the algorithm performed well on the real-world email exchange, computer, and brain networks. The mGH distance was used to successfully detect outlying shapes corresponding to events of significance. This suggests that the proposed algorithm may be useful for graph shape matching in various application domains.

%% file: text/8_acknowledgements.tex
\section*{Acknowledgements}
Part of this work was performed during the stay of Vladyslav Oles at Los Alamos National Laboratory, and he acknowledges the provided support. In addition, the authors would like to thank Facundo M\'emoli for detailed clarification of his work and useful comments, Amirali Kazeminejad for sharing preprocessed ABIDE I dataset, Bob Week for insightful conversations, Iana Khotsianivska for helping with figure design, Grace Grimm for reviewing and editing, and Luigi Boschetti for encouragement and inspiration.

%% file: text/9_appendix.tex
\begin{appendix}
	\section{Procedure \textproc{FindLarge$K$}}
	\label{FindLarge$K$}
	\hrulestart
	\begin{algorithmic}
		\Procedure{FindLarge$K$}{$\DX, d$} \Comment{$O(N^3)$}
		\Input{$\DX \in \R^{|X| \times |X|}$; $d > 0$}
		\Output{$d$-bounded $K \in \K_{\widetilde{M}(X, d)}(X)$}
		\State $K \gets \DX$
		\State $i \gets \textproc{FindLeastBoundedRow}(K, d)$ 
		\While{$i > 0$} \Comment{$i = 0$ if and only if $K$ is $d$-bounded}
		\State$K \gets K_{(i)(i)}$ \Comment{remove $i$-th row and column from $K$}
		\State$i \gets \textproc{FindLeastBoundedRow}(K, d)$  \Comment{decide which row to remove next} 
		\EndWhile
		\State \textbf{return} $K$ \Comment{$K$ is a $d$-bounded distance sample of $X$}
		\EndProcedure
	\end{algorithmic}
	\hruleend	
	
	\section{Procedure \textproc{FindLeastBoundedRow}}
	\label{FindLeastBoundedRow}
	\hrulestart
	\begin{algorithmic}
		\Procedure{FindLeastBoundedRow}{$A, d$} \Comment{$O(m^2)$}
		\Input{$A \in \R^{m \times m}$; $d > 0$}
		\Output{$\begin{cases}
			0, & \text{if $A$ is $d$-bounded} \\ 
			\min \left\{i \in \I{m}: \text{$\row{i}(A)$ is smallest least $d$-bounded}\right\}, & \text{otherwise}
			\end{cases}$}
		\Statex
		\State $i^* \gets 0$ 
		\State $n_{i^*} \gets 0$
		\State $s_{i^*} \gets 0$
		\For{$i = 1, \ldots, m$}
		\State $n_i \gets 0$
		\State $s_i \gets 0$
		\For{$j = 1, \ldots, m$} \Comment{\parbox[t]{0.4\linewidth}{count off-diagonal entries $<d$ and sum off-diagonal entries $\geq d$ in $\row{i}(A)$}}
		\If{$i \neq j$}
		\If{$A_{i,j} < d$}
		\State $n_i \gets n_i + 1$
		\Else
		\State $s_i \gets s_i + A_{i,j}$
		\EndIf
		\EndIf
		\EndFor
		\If{$n_i > n_{i^*}$ OR ($n_i = n_{i^*}$ AND $s_i < s_{i^*}$)} \Comment{\parbox[t]{0.35\linewidth}{choose smallest least $d$-bounded row from the first $i$ rows of $A$}}
		\State $i^* \gets i$
		\State $n_{i^*} \gets n_i$
		\State $s_{i^*} \gets s_i$
		\EndIf
		\EndFor
		\State \textbf{return} $i^*$ 
		\EndProcedure
	\end{algorithmic}
	\hruleend

\section{Procedure \textproc{SolveFeasibleAssignment}}
	\label{SolveFeasibleAssignment}
	\hrulestart
	\begin{algorithmic}
		\Procedure{SolveFeasibleAssignment}{$\mathbf{v}$, $\mathbf{u}$, $d$} \Comment{$O(q)$}
		\Input{$\mathbf{v} \in \R^p$ with entries $v_1 \leq \ldots \leq v_p$; $\mathbf{u} \in \R^q$ with entries $u_1 \leq \ldots \leq u_q$; $d > 0$}
		\Output{$\begin{cases}
			\text{TRUE}, & \text{if exists injective $f: \I{p} \to \I{q}$ s.t. $\left|v_t - u_{f(t)}\right| < d \quad \forall t \in \I{p}$} \\ 
			\text{FALSE}, & \text{otherwise}
			\end{cases}$}
		\Statex
		\State $l \gets 1$
		\For{$t = 1, \ldots, p$}
		\While{$|v_t - u_l| \geq d$} \Comment{find smallest available $u_l$ s.t. $|v_t - u_l| < d$}
		\State {$l \gets l + 1$}
		\If {$l > q$}
		\State \textbf{return} FALSE \Comment{no more available entries in $\mathbf{u}$}
		\EndIf
		\EndWhile
		\State $l \gets l + 1$ \Comment{assign $v_t$ to $u_l$, making the latter unavailable}
		\EndFor
		\State \textbf{return} TRUE	\Comment{$f$ can be constructed}
		\EndProcedure
	\end{algorithmic}
	\hruleend

\section{Procedure \textproc{CheckTheoremB}}
	\label{CheckTheoremB}
	\hrulestart
	\begin{algorithmic}
		\Procedure{CheckTheoremB}{$K,\DY,d$} \Comment{$O(N^3)$}
		\Input{$d$-bounded $K \in \K_n(X)$ for some $n \leq |Y|$; $\DY \in \R^{|Y| \times |Y|}$; $d > 0$}
		\Output{$\begin{cases}
			\text{TRUE}, & \text{if, for some $i \in \I{n}$, $\|\row{i}(K) - \row{i}(L)\|_\infty \geq d \quad \forall L \in \ps_n(\DY)$} \\ 
			\text{FALSE}, & \text{otherwise}
			\end{cases}$}
		\State $K \gets \textproc{SortEntriesInRows}(K)$ \Comment{sort entries in every row of $K$}
		\State $\DY \gets \textproc{SortEntriesInRows}(\DY)$ \Comment{sort entries in every row of $\DY$}
		\For{$i = 1, \ldots, n$}
		\State $\texttt{$i$\_satisfies} \gets$ TRUE
		\For{$j = 1, \ldots, |Y|$}
		\If{$\textproc{SolveFeasibleAssignment}(\row{i}(K), \row{j}(\DY), d)$}
		\State $\texttt{$i$\_satisfies} \gets$ FALSE \Comment{$\exists L \in \ps_n^{i\gets j}(\DY)\quad \|\row{i}(K) - \row{i}(L)\|_\infty < d$}
		\EndIf
		\EndFor
		\If{$\texttt{$i$\_satisfies}$}
		\State \textbf{return} TRUE \Comment{$\|\row{i}(K) - \row{i}(L)\|_\infty \geq d \quad \forall L \in \ps_n(\DY)$}
		\EndIf
		\EndFor
		\State \textbf{return} FALSE
		\EndProcedure
		\Statex
		\Procedure{SortEntriesInRows}{$A$} \Comment{$O(m^2 \log m)$}
		\Input{$A \in \R^{m \times m}$}
		\Output{$B \in \R^{m \times m}$ s.t. $\forall i \quad \row{i}(B)$ is a permutation of $\row{i}(A)$ and $B_{i,1} \leq \ldots \leq B_{i, m}$}
		\State \ldots
		\EndProcedure
	\end{algorithmic}
	\hruleend

\section{Procedure \textproc{VerifyLowerBound}}
	\label{VerifyLowerBound}
	\begin{algorithmic}
		\hrulestart
		\Procedure{VerifyLowerBound}{$\DX,\DY,d$} \Comment{$O(N^3)$}
		\Input{$\DX \in \R^{|X| \times |X|}$; $\DY \in \R^{|Y| \times |Y|}$; $d > 0$}
		\Output{$\begin{cases}
			\text{TRUE}, & \text{if verified that $\dmGH(X, Y) \geq \frac{d}{2}$} \\ 
			\text{FALSE}, & \text{if couldn't verify it}
			\end{cases}$}
		\State $K \gets $\textproc{FindLarge$K$}$(D^X, d)$ \Comment{$K \in \R^{n \times n}$}
		\State $L \gets $\textproc{FindLarge$K$}$(D^Y, d)$ \Comment{$L \in \R^{m \times m}$}
		\If{$n > |Y|$ OR $m > |X|$}
		\State \textbf{return} TRUE \Comment{\text{$\dmGH(X, Y) \geq \frac{d}{2}$ from Theorem A}}
		\ElsIf{$\textproc{CheckTheoremB}(K, D^Y, d)$ OR $\textproc{CheckTheoremB}(L, D^X, d)$}
		
		\State \textbf{return} TRUE \Comment{\text{$\dmGH(X, Y) \geq \frac{d}{2}$ from Theorem B}}
		\Else
		\State \textbf{return} FALSE
		\EndIf
		\EndProcedure
		\hruleend
	\end{algorithmic}

\section{Procedure \textproc{FindLowerBound}}
	\label{FindLowerBound}
	\begin{algorithmic}
		\hrulestart
		\Procedure{FindLowerBound}{$\DX,\DY$} \Comment{$O(N^7)$}
		\Input{$\DX \in \R^{|X|\times |X|}$; $\DY \in \R^{|Y|\times |Y|}$}
		\Output{$b_L \in \R$ s.t. $\dmGH(X, Y) \geq b_L$}
		\State $\Delta \gets \emptyset$
		\For{$i = 1, \ldots, |X|$} \Comment{construct $\Delta$}
		\For{$j = 1, \ldots, |X|$}
		\For{$k = 1, \ldots, |Y|$}
		\For{$l = 1, \ldots, |Y|$}
		\State $\Delta \gets \Delta \cup |\DX_{i,j} - \DY_{k,l}|$
		\EndFor
		\EndFor
		\EndFor
		\EndFor
		\State $\{\delta_i\}_{i=1}^{|\Delta|} \gets \textproc{SortSet}(\Delta)$
		\For{$i = |\Delta|, \ldots, 1$} \Comment{find largest $\frac{\delta_i}{2}$ for which the answer is ``yes''}
		\If{$\textproc{VerifyLowerBound}(\DX, \DY, \delta_i)$}
		\State \textbf{return} $\frac{\delta_i}{2}$
		\EndIf
		\EndFor
		\State \textbf{return} $0$
		\EndProcedure
		\Statex
		\Procedure{SortSet}{$T$} \Comment{$O(|T|\log |T|)$}
		\Input{$T \subset \R$}
		\Output{$\{t_i\}_{i=1}^{|T|}$, where $t_1 < \ldots < t_{|T|}$ and $t_i \in T \quad \forall i \in \I{|T|}$}
		\Statex
		\State \ldots
		\EndProcedure
		\hruleend
	\end{algorithmic}

\section{Procedure \textproc{SampleSmallDistortion}}
	\label{SampleSmallDistortion}
	\begin{algorithmic}
		\hrulestart
		\Procedure{SampleSmallDistortion}{$\DX, \DY$} \Comment{$O(N^3)$}
		\Input{$\DX \in \R^{|X|\times |X|}$ s.t. $\DX_{i,j} = d_X(x_i, x_j)$; $\DY \in \R^{|Y|\times |Y|}$ s.t. $\DY_{i,j} = d_Y(y_i, y_j)$}
		\Output{$\dis R_\varphi$, where $R_\varphi = \left\{(x, \varphi(x)); x \in X \right\}$ for some $\varphi:X \to Y$}
		\State $\pi \gets \textproc{GetRandomOrder}(|X|)$ \Comment{choose order in which to map the points in $X$}
		\State $R_\varphi^{(0)} \gets \emptyset$
		\For{$i = 1, \ldots, |X|$} \Comment{map $x_{\pi(i)}$ at step $i$}
		\State $\delta^{(i)} \gets \infty$ \Comment{initialize $\dis R_\varphi^{(i)}$}
		\State $j_i \gets 0$
		\For{$j = 1, \ldots, |Y|$} \Comment{find $y_{j_i}$ that minimizes $\dis \Big(R_\varphi^{(i-1)} \cup \big\{(x_{\pi(i)}, y_{j_i})\big\}\Big)$}
		\State $\delta \gets \delta^{(i-1)}$ 
		\For{$k = 1, \ldots, i-1$}\Comment{find $\dis \Big(R_\varphi^{(i-1)} \cup \big\{(x_{\pi(i)}, y_j)\big\}\Big)$}
		\If{$\delta < \big|\DX_{\pi(i), \pi(k)} - \DY_{j, j_k}\big|$}
		\State $\delta \gets \big|\DX_{\pi(i), \pi(k)} - \DY_{j, j_k}\big|$
		\EndIf
		\EndFor
		\If{$\delta < \delta^{(i)}$}\Comment{choose better $y_{j_i}$ candidate}
		\State $\delta^{(i)} \gets \delta$
		\State $j_i \gets j$
		\EndIf
		\EndFor
		\State $R_\varphi^{(i)} \gets R_\varphi^{(i-1)} \cup \{(x_{\pi(i)}, y_{j_i})\}$
		\EndFor
		\State \textbf{return} $\delta^{(|X|)}$
		\EndProcedure
		\Statex
		\Procedure{GetRandomOrder}{$n$} \Comment{$O(n)$}
		\Input{$n \in \N$}
		\Output{randomly generated permutation $\pi$ of $\I{n}$}
		\State \ldots
		\EndProcedure
		\hruleend
	\end{algorithmic}

\section{Procedure \textproc{FindUpperBound}}
	\label{FindUpperBound}
	\begin{algorithmic}
		\hrulestart
		\Procedure{FindUpperBound}{$\DX, \DY$} \Comment{\parbox[t]{.31\linewidth}{$O(sN^3)$, where $s$ is the total number of sampled mappings}}
		\Input{$\DX \in \R^{|X|\times |X|}$; $\DY \in \R^{|Y|\times |Y|}$}
		\Output{$b_U \in \R$ s.t. $\dmGH(X, Y) \leq b_U$}
		\State $\dis \varphi^* \gets \infty$
		\State $s \gets \textproc{DecideSampleSize}(|X|, |Y|)$ \Comment{number of mappings from $X \to Y$ to sample}
		\For{$i = 1, \ldots, s$} \Comment{randomized sampling from $X \to Y$}
		\State $\dis \varphi \gets \textproc{SampleSmallDistortion}(\DX, \DY)$
		\If{$\dis \varphi < \dis \varphi^*$}
		\State $\dis \varphi^* \gets \dis \varphi$
		\EndIf
		\EndFor
		\State $\dis \psi^* \gets \infty$
		\State $s \gets \textproc{DecideSampleSize}(|Y|, |X|)$ \Comment{number of mappings from $Y \to X$ to sample}
		\For{$i = 1, \ldots, s$} \Comment{randomized sampling from $Y \to X$}
		\State $\dis \psi \gets \textproc{SampleSmallDistortion}(\DY, \DX)$
		\If{$\dis \psi < \dis \psi^*$}
		\State $\dis \psi^* \gets \dis \psi$
		\EndIf
		\EndFor
		\State \textbf{return} $\frac{1}{2}\max \{\dis \varphi^*, \dis \psi^*\}$ 
		\EndProcedure
		\Statex
		\Procedure{DecideSampleSize}{$n, m$} \Comment{$O(1)$}
		\Input{$n \in \N$; $m \in \N$}
		\Output{sample size $s \in \N$ for the population size of $m^n$}
		\State $\ldots$
		\EndProcedure
		\hruleend
	\end{algorithmic}

\end{appendix}